\documentclass[journal]{IEEEtran}
\PassOptionsToPackage{table}{xcolor}
\usepackage{amsmath,amsfonts}
\usepackage{graphicx}
\usepackage{cite}
\usepackage{amssymb}
\usepackage{booktabs}
\usepackage{enumerate}
\usepackage{pgfplots}
\usepackage{mathtools}
\usepackage{pstool}
\usepackage{psfrag}
\usepackage[]{algpseudocode}
\newcommand{\algmargin}{\the\ALG@thistlm}
\usepackage[ruled,vlined,linesnumbered,lined,commentsnumbered]{algorithm2e}
\usepackage[printonlyused,withpage]{acronym}
\usepackage[bookmarks=false]{hyperref}
\usepackage{caption}
\usepackage{subcaption}
\usepackage{xcolor}
\usepackage[bottom]{footmisc}

\SetKwInOut{Input}{Input}
\SetKwInOut{Output}{Output}

\pgfplotsset{compat=1.17}
\acrodef{KPI}[KPI]{key performance indicator}
\acrodef{WRT}[w.r.t.]{with respect to}
\acrodef{WP}[w.p.]{with probability}
\acrodef{LHS}[l.h.s.]{left hand side}
\acrodef{RHS}[r.h.s.]{right hand side}

\acrodef{2D}[2D]{two-dimensional}
\acrodef{3D}[3D]{three-dimensional}
\acrodef{3GPP}[3GPP]{3rd Generation Partnership Project}
\acrodef{AOA}[AOA]{angle-of-arrival}
\acrodef{AWGN}[AWGN]{additive white Gaussian noise}
\acrodef{BS}[BS]{base station}
\acrodef{BFT}[BFT]{beam focusing and tracking}
\acrodef{CFR}[CFR]{channel frequency response}
\acrodef{CIR}[CIR]{channel impulse response}
\acrodef{CTRV}[CTRV]{\textit{Constant Turn Rate and Velocity}}
\acrodef{DDPG}[DDPG]{deep deterministic policy gradient}
\acrodef{DRL}[DRL]{deep reinforcement learning}
\acrodef{ELAA}[ELAA]{extremely large aperture array}
\acrodef{FF}[FF]{far-field}
\acrodef{FR3}[FR3]{frequency range~3}
\acrodef{ISAC}[ISAC]{integrated sensing and communication}
\acrodef{KF}[KF]{Kalman filter}
\acrodef{LMS}[LMS]{least means square}

\acrodef{LOS}[LOS]{line-of-sight}
\acrodef{LS}[LS]{least squares}
\acrodef{LSTM}[LSTM]{long short-term memory}
\acrodef{MAP}[MAP]{maximum a posteriori}
\acrodef{MC}[MC]{Monte Carlo}
\acrodef{MDP}[MDP]{Markov decision process}
\acrodef{MIMO}[MIMO]{multiple-input multiple-output}
\acrodef{ML}[ML]{maximum likelihood}
\acrodef{MLE}[MLE]{maximum likelihood estimation}
\acrodef{MMSE}[MMSE]{minimum mean square error}
\acrodef{MSE}[MSE]{mean square error}
\acrodef{MUI}[MUI]{multi-user interference}
\acrodef{MUSIC}[MUSIC]{multiple signal classification}
\acrodef{NF}[NF]{near-field}
\acrodef{NLOS}[NLOS]{non-line-of-sight}
\acrodef{NN}[NN]{neural network}
\acrodef{OU}[OU]{Ornstein-Uhlenbeck}
\acrodef{PDP}[PDP]{power delay profile}
\acrodef{RCS}[RCS]{radar cross section}
\acrodef{RF}[RF]{radio frequency}
\acrodef{RFID}[RFID]{radio frequency identification}
\acrodef{RFS}[RFS]{random finite set}
\acrodef{RI}[RI]{range information}
\acrodef{RL}[RL]{reinforcement learning}
\acrodef{RMSE}[RMSE]{root-mean-square error}
\acrodef{ReNN}[regression NN]{regression neural network}
\acrodef{ROI}[ROI]{region of interest}
\acrodef{RTT}[RTT]{round-trip time}
\acrodef{RV}[RV]{random variable}
\acrodef{SI}[SI]{soft information}
\acrodef{SIR}[SIR]{signal-to-interference ratio}
\acrodef{SLAM}[SLAM]{simultaneous localization and mapping}
\acrodef{SMC}[SMC]{sequential monte carlo}
\acrodef{SNR}[SNR]{signal-to-noise ratio}
\acrodef{SINR}[SINR]{signal-to-interference-plus-noise ratio}
\acrodef{TDOA}[TDOA]{time difference-of-arrival}
\acrodef{TOA}[TOA]{time-of-arrival}
\acrodef{TOF}[TOF]{time-of-flight}
\acrodef{UAV}[UAV]{unmanned aerial vehicles}
\acrodef{UE}[UE]{user equipment}
\acrodef{UKF}[UKF]{unscented Kalman filter}
\acrodef{UWB}[UWB]{ultra-wideband}
\acrodef{WLS}[WLS]{weighted least squares}
\acrodef{xG}[xG]{Next-generation}
\acrodef{AFDM}[AFDM]{affine frequency division multiplexing}
\acrodef{PD}[PD]{photodetector}
\acrodef{LO}[LO]{local oscillator}
\acrodef{w.r.t.}[w.r.t.]{with respect to}
\acrodef{AFT}[AFT]{affine Fourier transform}
\acrodef{RAQRs}[RAQRs]{Rydberg atomic quantum receivers}
\acrodef{RAQR}[RAQR]{Rydberg atomic quantum receiver}
\acrodef{DAFT}[DAFT]{discrete affine Fourier transform}
\acrodef{NLS}[NLS]{nonlinear least square}
\acrodef{FFT}[FFT]{fast Fourier transform}
\acrodef{ICI}[ICI]{intercarrier interference}
\acrodef{OTFS}[OTFS]{orthogonal time frequency space}
\acrodef{OFDM}[OFDM]{orthogonal frequency division multiplexing}
\acrodef{DD}[DD]{delay-Doppler}
\acrodef{SQL}[SQL]{standard quantum limit}
\acrodef{IDAFT}[IDAFT]{inverse DAFT}
\acrodef{NRMSE}[NRMSE]{normalized root mean square error}
\acrodef{RIS}[RIS]{reconfigurable intelligent surface}
\acrodef{LEO}[LEO]{low earth orbit}
\acrodef{MC-AFDM}[MC-AFDM]{multi-chirp affine frequency division multiplexing}
\acrodef{SC-AFDM}[SC-AFDM]{single-chirp AFDM}
\acrodef{OMP}[OMP]{orthogonal matching pursuit}
\acrodef{FIM}[FIM]{Fisher information matrix}
\acrodef{PDF}[PDF]{probability density function}
\acrodef{EFIM}[EFIM]{equivalent Fisher information matrix}
\acrodef{CRLB}[CRLB]{Cram{\'e}r-Rao lower bound}
\acrodef{DC-AFDM}[DC-AFDM]{dual-chirp AFDM}
\acrodef{DD}[DD]{doubly-dispersive}

\begin{document}
\newtheorem{theorem}{\bf Theorem}
\newtheorem{acknowledgement}[theorem]{Acknowledgement}
\newtheorem{axiom}[theorem]{Axiom}
\newtheorem{case}[theorem]{Case}
\newtheorem{claim}[theorem]{Claim}
\newtheorem{conclusion}[theorem]{Conclusion}
\newtheorem{condition}[theorem]{Condition}
\newtheorem{conjecture}[theorem]{Conjecture}
\newtheorem{criterion}[theorem]{Criterion}
\newtheorem{definition}{Definition}
\newtheorem{exercise}[theorem]{Exercise}
\newtheorem{lemma}{Lemma}
\newtheorem{corollary}{Corollary}
\newtheorem{notation}[theorem]{Notation}
\newtheorem{problem}[theorem]{Problem}
\newtheorem{proposition}{Proposition}
\newtheorem{solution}[theorem]{Solution}
\newtheorem{summary}[theorem]{Summary}
\newtheorem{assumption}{Assumption}
\newtheorem{example}{\bf Example}
\newtheorem{remark}{\bf Remark}

\newtheorem{thm}{Corollary}[section]
\renewcommand{\thethm}{\arabic{section}.\arabic{thm}}

\def\qed{$\Box$}
\def\QED{\mbox{\phantom{m}}\nolinebreak\hfill$\,\Box$}
\def\proof{\noindent{\emph{Proof:} }}
\def\poof{\noindent{\emph{Sketch of Proof:} }}
\def
\endproof{\hspace*{\fill}~\qed
\par
\endtrivlist\unskip}
\def\endproof{\hspace*{\fill}~\qed\par\endtrivlist\vskip3pt}

\def\E{\mathsf{E}}
\def\eps{\varepsilon}
\def\Lsp{{\boldsymbol L}}
\def\Bsp{{\boldsymbol B}}
\def\lsp{{\boldsymbol\ell}}
\def\Ltsp{{\Lsp^2}}
\def\Lpsp{{\Lsp^p}}
\def\Linsp{{\Lsp^{\infty}}}
\def\LtR{{\Lsp^2(\Rst)}}
\def\ltZ{{\lsp^2(\Zst)}}
\def\ltsp{{\lsp^2}}
\def\ltZt{{\lsp^2(\Zst^{2})}}
\def\ninN{{n{\in}\Nst}}
\def\oh{{\frac{1}{2}}}
\def\grass{{\cal G}}
\def\ord{{\cal O}}
\def\dist{{d_G}}
\def\conj#1{{\overline#1}}
\def\ntoinf{{n \rightarrow \infty}}
\def\toinf{{\rightarrow \infty}}
\def\tozero{{\rightarrow 0}}
\def\trace{{\operatorname{trace}}}
\def\ord{{\cal O}}
\def\UU{{\cal U}}
\def\rank{{\operatorname{rank}}}
\def\acos{{\operatorname{acos}}}

\def\SINR{\mathsf{SINR}}
\def\SNR{\mathsf{SNR}}
\def\SIR{\mathsf{SIR}}
\def\tSIR{\widetilde{\mathsf{SIR}}}
\def\Ei{\mathsf{Ei}}
\def\l{\left}
\def\r{\right}
\def\lb{\left\{}
\def\rb{\right\}}

\setcounter{page}{1}

\newcommand{\eref}[1]{(\ref{#1})}
\newcommand{\fig}[1]{Fig.\ \ref{#1}}

\def\bydef{:=}
\def\ba{{\mathbf{a}}}
\def\bb{{\mathbf{b}}}
\def\bc{{\mathbf{c}}}
\def\bd{{\mathbf{d}}}
\def\bee{{\mathbf{e}}}
\def\bff{{\mathbf{f}}}
\def\bg{{\mathbf{g}}}
\def\bh{{\mathbf{h}}}
\def\bi{{\mathbf{i}}}
\def\bj{{\mathbf{j}}}
\def\bk{{\mathbf{k}}}
\def\bl{{\mathbf{l}}}
\def\bm{{\mathbf{m}}}
\def\bn{{\mathbf{n}}}
\def\bo{{\mathbf{o}}}
\def\bp{{\mathbf{p}}}
\def\bq{{\mathbf{q}}}
\def\br{{\mathbf{r}}}
\def\bs{{\mathbf{s}}}
\def\bt{{\mathbf{t}}}
\def\bu{{\mathbf{u}}}
\def\bv{{\mathbf{v}}}
\def\bw{{\mathbf{w}}}
\def\bx{{\mathbf{x}}}
\def\by{{\mathbf{y}}}
\def\bz{{\mathbf{z}}}
\def\b0{{\mathbf{0}}}

\def\bA{{\mathbf{A}}}
\def\bB{{\mathbf{B}}}
\def\bC{{\mathbf{C}}}
\def\bD{{\mathbf{D}}}
\def\bE{{\mathbf{E}}}
\def\bF{{\mathbf{F}}}
\def\bG{{\mathbf{G}}}
\def\bH{{\mathbf{H}}}
\def\bI{{\mathbf{I}}}
\def\bJ{{\mathbf{J}}}
\def\bK{{\mathbf{K}}}
\def\bL{{\mathbf{L}}}
\def\bM{{\mathbf{M}}}
\def\bN{{\mathbf{N}}}
\def\bO{{\mathbf{O}}}
\def\bP{{\mathbf{P}}}
\def\bQ{{\mathbf{Q}}}
\def\bR{{\mathbf{R}}}
\def\bS{{\mathbf{S}}}
\def\bT{{\mathbf{T}}}
\def\bU{{\mathbf{U}}}
\def\bV{{\mathbf{V}}}
\def\bW{{\mathbf{W}}}
\def\bX{{\mathbf{X}}}
\def\bY{{\mathbf{Y}}}
\def\bZ{{\mathbf{Z}}}

\def\bxi{{\boldsymbol{\xi}}}

\def\sT{{\mathsf{T}}}
\def\sH{{\mathsf{H}}}
\def\cmp{{\text{cmp}}}
\def\cmm{{\text{cmm}}}
\def\WPT{{\text{WPT}}}
\def\lo{{\text{lo}}}
\def\gl{{\text{gl}}}

\def\tT{{\widetilde{T}}}
\def\tF{{\widetilde{F}}}
\def\tP{{\widetilde{P}}}
\def\tG{{\widetilde{G}}}
\def\tbh{{\widetilde{\mathbf{h}}}}
\def\tbg{{\widetilde{\mathbf{g}}}}

\def\mA{{\mathbb{A}}}
\def\mB{{\mathbb{B}}}
\def\mC{{\mathbb{C}}}
\def\mD{{\mathbb{D}}}
\def\mE{{\mathbb{E}}}
\def\mF{{\mathbb{F}}}
\def\mG{{\mathbb{G}}}
\def\mH{{\mathbb{H}}}
\def\mI{{\mathbb{I}}}
\def\mJ{{\mathbb{J}}}
\def\mK{{\mathbb{K}}}
\def\mL{{\mathbb{L}}}
\def\mM{{\mathbb{M}}}
\def\mN{{\mathbb{N}}}
\def\mO{{\mathbb{O}}}
\def\mP{{\mathbb{P}}}
\def\mQ{{\mathbb{Q}}}
\def\mR{{\mathbb{R}}}
\def\mS{{\mathbb{S}}}
\def\mT{{\mathbb{T}}}
\def\mU{{\mathbb{U}}}
\def\mV{{\mathbb{V}}}
\def\mW{{\mathbb{W}}}
\def\mX{{\mathbb{X}}}
\def\mY{{\mathbb{Y}}}
\def\mZ{{\mathbb{Z}}}

\def\cA{\mathcal{A}}
\def\cB{\mathcal{B}}
\def\cC{\mathcal{C}}
\def\cD{\mathcal{D}}
\def\cE{\mathcal{E}}
\def\cF{\mathcal{F}}
\def\cG{\mathcal{G}}
\def\cH{\mathcal{H}}
\def\cI{\mathcal{I}}
\def\cJ{\mathcal{J}}
\def\cK{\mathcal{K}}
\def\cL{\mathcal{L}}
\def\cM{\mathcal{M}}
\def\cN{\mathcal{N}}
\def\cO{\mathcal{O}}
\def\cP{\mathcal{P}}
\def\cQ{\mathcal{Q}}
\def\cR{\mathcal{R}}
\def\cS{\mathcal{S}}
\def\cT{\mathcal{T}}
\def\cU{\mathcal{U}}
\def\cV{\mathcal{V}}
\def\cW{\mathcal{W}}
\def\cX{\mathcal{X}}
\def\cY{\mathcal{Y}}
\def\cZ{\mathcal{Z}}
\def\cd{\mathcal{d}}
\def\Mt{M_{t}}
\def\Mr{M_{r}}
\def\O{\Omega_{M_{t}}}
\newcommand{\figref}[1]{{Fig.}~\ref{#1}}
\newcommand{\tabref}[1]{{Table}~\ref{#1}}

\newcommand{\fb}{\tx{fb}}
\newcommand{\nf}{\tx{nf}}
\newcommand{\BC}{\tx{(bc)}}
\newcommand{\MAC}{\tx{(mac)}}
\newcommand{\Pout}{p_{\mathsf{out}}}
\newcommand{\nnn}{\nn\\}
\newcommand{\FB}{\tx{FB}}
\newcommand{\TX}{\tx{TX}}
\newcommand{\RX}{\tx{RX}}
\renewcommand{\mod}{\tx{mod}}
\newcommand{\m}[1]{\mathbf{#1}}
\newcommand{\td}[1]{\tilde{#1}}
\newcommand{\sbf}[1]{\scriptsize{\textbf{#1}}}
\newcommand{\stxt}[1]{\scriptsize{\textrm{#1}}}
\newcommand{\suml}[2]{\sum\limits_{#1}^{#2}}
\newcommand{\sumlk}{\sum\limits_{k=0}^{K-1}}
\newcommand{\eqhsp}{\hspace{10 pt}}
\newcommand{\tx}[1]{\texttt{#1}}
\newcommand{\Hz}{\ \tx{Hz}}
\newcommand{\sinc}{\tx{sinc}}
\newcommand{\diag}{\mathrm{diag}}
\newcommand{\MAI}{\tx{MAI}}
\newcommand{\ISI}{\tx{ISI}}
\newcommand{\IBI}{\tx{IBI}}
\newcommand{\CN}{\tx{CN}}
\newcommand{\CP}{\tx{CP}}
\newcommand{\ZP}{\tx{ZP}}
\newcommand{\ZF}{\tx{ZF}}
\newcommand{\SP}{\tx{SP}}
\newcommand{\MMSE}{\tx{MMSE}}
\newcommand{\MINF}{\tx{MINF}}
\newcommand{\RC}{\tx{MP}}
\newcommand{\MBER}{\tx{MBER}}
\newcommand{\MSNR}{\tx{MSNR}}
\newcommand{\MCAP}{\tx{MCAP}}
\newcommand{\vol}{\tx{vol}}
\newcommand{\ah}{\hat{g}}
\newcommand{\tg}{\tilde{g}}
\newcommand{\teta}{\tilde{\eta}}
\newcommand{\heta}{\hat{\eta}}
\newcommand{\uh}{\m{\hat{s}}}
\newcommand{\eh}{\m{\hat{\eta}}}
\newcommand{\hv}{\m{h}}
\newcommand{\hh}{\m{\hat{h}}}
\newcommand{\Po}{P_{\mathrm{out}}}
\newcommand{\Poh}{\hat{P}_{\mathrm{out}}}
\newcommand{\Ph}{\hat{\gamma}}
\newcommand{\mat}[1]{\begin{matrix}#1\end{matrix}}
\newcommand{\ud}{^{\dagger}}
\newcommand{\C}{\mathcal{C}}
\newcommand{\nn}{\nonumber}
\newcommand{\nInf}{U\rightarrow \infty}

\title{Multi-Chirp AFDM for Rydberg Atomic Quantum Receivers: Waveform and Algorithm Design\!
\thanks{The fundamental research described in this paper was supported by the National Research Foundation (NRF) of Korea under Grant RS-2024-00409492, and by the German Research Foundation (DFG) through the QUBYSM Project with Grant No. G:(GEPRIS)576171458. \\ \textit{(Corresponding author: Sunwoo Kim)}
\newline \indent Hanvit Kim, Kihong Min, and Sunwoo Kim are with the Department of Electronics and Computer Engineering, Hanyang University, Seoul, 04763, South Korea (email: dante0813@hanyang.ac.kr; khmin705@hanyang.ac.kr; remero@hanyang.ac.kr).

Hyeon Seok Rou and Giuseppe Thadeu Freitas de Abreu are with the School of Computer Science Engineering, Constructor University, 28579 Bremen, Germany (email: hrou@constructor.university; gabreu@constructor.university).}
}
\author{
\IEEEauthorblockN{Hanvit Kim,~\IEEEmembership{Graduate Student Member,~IEEE}, Hyeon Seok Rou,~\IEEEmembership{Member,~IEEE}, \\ Kihong Min,~\IEEEmembership{Graduate Student Member,~IEEE}, Giuseppe Thadeu Freitas de Abreu,~\IEEEmembership{Senior Member,~IEEE}, \\ and Sunwoo Kim,~\IEEEmembership{Senior Member,~IEEE}}
%
%
\vspace{-4ex}
}
\maketitle

\begin{abstract}
We propose a \ac{MC-AFDM} scheme for joint delay-Doppler estimation with \ac{RAQRs}. 
The work is motivated by the fact that \ac{RAQRs}, while offering superior sensitivity and advantageous sensing capabilities, suffer from an optical ambiguity due to Doppler shifts in \ac{DD} channel caused by target mobility, which precludes the reliable estimation of delay-Doppler parameters.
To resolve this optical ambiguity and unleash the potential of RAQRs in \ac{DD} channel, the proposed \ac{MC-AFDM} employs multiple distinct AFDM post-chirp signals to overcome the rank-deficiency problem of the classical \ac{SC-AFDM}, thereby enabling accurate delay-Doppler estimation of multiple targets. 
Our analysis reveals that the edge distribution of the multiple post-chirp parameters can further improve estimation accuracy by minimizing the condition number.
Building on the proposed \ac{MC-AFDM} waveform, we design a sequential signal processing algorithm based on \ac{OMP} and \ac{LS}, and we derive the theoretical lower bounds for delay and Doppler estimation.
Numerical results show that the proposed \ac{MC-AFDM} improves range and velocity estimation accuracy by up to two orders of magnitude compared to the classical \ac{SC-AFDM}, and approaches its theoretical bounds through post-chirp optimization—validating the quantum-induced advantage of \ac{RAQRs} for high-resolution quantum wireless sensing.
\end{abstract}


\acresetall	

\section{Introduction}\label{Sec-I}

\IEEEPARstart{R}{ecently}, \ac{RAQRs} have attracted extensive attention as a key enabler for next-generation wireless communications~\cite{jing2020atomic, du2022realization, cui2026rydberg, cui_MIMO, gong2025rydberg,kang2026deepq, gong2026rydberg, zhu2025general, cui2025rydberg_frontier, chen2025new, peng2025ris, peng2026enhanced}.
By replacing the bulky \ac{RF} components with the optical equipment, \ac{RAQRs} can achieve extremely high-level electric-field sensitivity, approaching the level of the \ac{SQL} ($\sim 700 \mathrm{pV} \cdot \mathrm{~cm}^{-1} \cdot \mathrm{~Hz}^{-1 / 2}$)~\cite{jing2020atomic}.
Furthermore, the numerous energy levels of the Rydberg atom enable the simultaneous detection of a broad range of multi-band signals, spanning the MHz-to-THz range~\cite{du2022realization, cui2026rydberg}.
To fully unleash these exceptional benefits, the early research on \ac{RAQRs} has started with analyzing the interaction between Rydberg atoms and \ac{RF} signals by quantum optics, and modeling the received signal system model~\cite{cui_MIMO,gong2026rydberg, zhu2025general}.
Thereafter, integrating \ac{RAQRs} with several advanced communications technologies has been studied, such as atomic \ac{MIMO}~\cite{cui_MIMO,gong2025rydberg,kang2026deepq}, quantum \ac{ISAC}~\cite{chen2025new}, atomic \ac{RIS}~\cite{peng2025ris}, and \ac{LEO} satellite communications~\cite{peng2026enhanced}.

Along with wireless communications based on \ac{RAQRs}, wireless sensing with \ac{RAQRs}, also referred to as quantum wireless sensing, is another branch that has significantly advanced through recent efforts~\cite{simons2019rydberg, Quantum_Wireless, holloway2020multiple, kim2025quantum, gong2025rydberg_TVT,cui2025realizing,11168279, kim2025quantum_2, Chen_2026_JSAC, Guo_AoA}. 
The extreme sensitivity and low noise level of \ac{RAQRs} significantly increase the sensing accuracy compared to their classical counterparts.
The early studies on quantum wireless sensing have started with demonstration with a testbed prototype, including phase detection~\cite{simons2019rydberg}, micro-vibration detection~\cite{Quantum_Wireless}, and multi-band \ac{RF} signal detection~\cite{holloway2020multiple}.
Thereafter, various signal processing algorithms for quantum wireless sensing have been proposed, such as \ac{AOA} estimation~\cite{kim2025quantum,gong2025rydberg_TVT}, range estimation~\cite{cui2025realizing}, and localization~\cite{11168279}.
Furthermore, several \ac{RAQRs}-enabled sensing paradigms and receiver architectures have been proposed, including multi-band quantum wireless sensing~\cite{kim2025quantum_2} and \ac{AOA} estimation with a single atomic receiver~\cite{Chen_2026_JSAC, Guo_AoA}.

Despite these recent efforts, it is obvious that most of the existing works on \ac{RAQRs} have focused on the receiver side technology, such as receiver architecture design~\cite{wang2025multi,chen2026wideband,chen2025harnessing}.
Meanwhile, the transmitter side technology for \ac{RAQRs} is still at a nascent stage, which is essential for the realization of optimal communications and sensing performance.
This lack of studies has driven researchers to investigate the new transmission technology that is compatible with \ac{RAQRs}.
For example, authors in~\cite{cui2024mimo} proposed a precoding technique for \ac{RAQRs} to achieve the enhanced achievable rates and capacity in atomic \ac{MIMO} systems.
Also, authors in~\cite{cui2025realizing} proposed a self-heterodyne sensing scheme, which removes the necessity of extra \ac{LO} by utilizing the signal transmitter itself as a reference signal transmission.
By utilizing this self-heterodyne sensing scheme, the estimation accuracy and the range of the target have tremendously increased compared to the classical \ac{RAQRs} without exploiting the external \ac{LO}.
%

Motivated by these recent advances in transmitter side technology for \ac{RAQRs}, this paper focuses on a waveform design for \ac{RAQRs}, which, to the best of our knowledge, has not been addressed yet.
In particular, we focus on the design of a \ac{MC-AFDM} waveform for accurate joint delay-Doppler estimation with \ac{RAQRs}.
While \ac{AFDM} has been known as a promising candidate for 6G and beyond-6G communications due to its robustness against delay-Doppler channels and high flexibility~\cite{Rou_Chirp26, Rou_Mag_2, Rou_SPM24}, its utilization for \ac{RAQRs} is challenging due to several factors.
These technical challenges include: \textit{1)} the absence of a quantum system model, framework, and compatible signal processing algorithm for \ac{AFDM}-Rydberg atom cross interaction, and \textit{2)} delay-Doppler measurement ambiguity introduced by the unique optical readout of Rydberg probes in a \ac{DD} channel.
Finding the solution to these challenges is crucial for the utilization of \ac{AFDM} to the \ac{RAQRs}, which will be a key enabler for future quantum communications and quantum wireless sensing.

The new \ac{MC-AFDM} waveform, designed for accurate joint delay-Doppler estimation with \ac{RAQRs}, utilizes multiple and distinct post-chirps for \ac{AFDM} signal transmission, and therefore differs from the conventional \ac{SC-AFDM} that utilizes a single and unique post-chirp (generally denoted as $c_{1}$)~\cite{AFDM_Tek}.
By transmitting \ac{AFDM} with multiple post-chirps, the optical measurement of \ac{RAQRs} for the delay-Doppler estimation becomes distinct across each time frame with different post-chirps, enabling the unique delay-Doppler estimation.
In our previous study~\cite{kim2026dual}, we have found that two-distinct post-chirps enable the delay-Doppler estimation with \ac{RAQRs} under a single target scenario.

In so doing, the paper extends the previous work~\cite{kim2026dual} to a multi-target scenario and shows the efficiency of \ac{MC-AFDM} under the \ac{DD} channel.
Specifically, we propose the chirp parameter design criteria of the \ac{MC-AFDM} to be compatible with \ac{RAQRs} and enhance the delay-Doppler estimation accuracy.
Also, we propose an \ac{OMP}-based signal processing algorithm and derive the theoretical lower bound for delay-Doppler estimation of the proposed \ac{MC-AFDM} with \ac{RAQRs}. 

Our key contributions are summarized in the following:
\begin{itemize}
   \item \textbf{Rydberg-\ac{AFDM} quantum system model analysis}: We analyze the quantum system model that describes the interaction between the Rydberg atom and \ac{AFDM} signal.
   Also, the self-heterodyne sensing-based delay-Doppler estimation framework for the \ac{AFDM} is designed.
   In this framework, we reveal that the single and unique post-chirp of classical \ac{SC-AFDM} signal introduces a unique optical measurement for delay-Doppler estimation with \ac{RAQRs}, which blocks the joint delay-Doppler estimation of the multi-target;
   \item \textbf{\ac{MC-AFDM} waveform design}: The \ac{MC-AFDM} is proposed to eliminate the inherent ambiguity induced by the unique optical measurement of \ac{SC-AFDM}.
   The proposed \ac{MC-AFDM} resolves the rank-deficiency problem of the \ac{SC-AFDM} by utilizing distinct post-chirps within multiple different time-frames, enabling the unique delay-Doppler estimation.
   Furthermore, we propose the chirp-parameter design criteria for the proposed \ac{MC-AFDM} to further enhance the joint delay-Doppler estimation performance; and
   \item \textbf{Delay-Doppler estimation algorithm design and lower bound derivation}: We propose a signal processing algorithm for the joint delay-Doppler estimation with the proposed \ac{MC-AFDM}.
   The proposed algorithm comprises two sequential steps, including \ac{OMP}-based fluctuation frequency estimation and \ac{LS}-based joint delay-Doppler estimation.
   Furthermore, we derive the theoretical lower bound of delay-Doppler estimation and demonstrate that the estimation accuracy of the proposed algorithm approaches its theoretical lower bound.
\end{itemize}

The remainder of the paper is organized as follows. 
The physics and framework for \ac{RAQRs} for \ac{AFDM} detection and processing are analyzed in Section~\ref{Sec-II}. 
The measured signal system model based on the established framework with the problem analysis is proposed in Section~\ref{Sec-III}.
In Section~\ref{Sec-IV}, the \ac{MC-AFDM} and its design criteria are proposed.
Section~\ref{Sec-V} provides the algorithm and the lower bound for joint delay-Doppler estimation with the proposed \ac{MC-AFDM}.
Eventually, numerical experiments and discussions are provided in Section~\ref{Sec-VI}, and we conclude this paper in Section~\ref{Sec-VII}.

$\textit{Notation:}$ All scalars are represented by upper or lowercase letters, while column vectors and matrices are denoted by bold lowercase and uppercase letters, respectively.
$(\cdot)^{\mathsf{T}}$,$(\cdot)^{-1}$, and $(\cdot)^{\mathsf{H}}$ denote transpose, inverse, and conjugate transpose operators, respectively.
The operators $\mathrm{diag}\{\cdot\}$ and $\mathrm{blkdiag}\{\cdot\}$ denote the diagonal matrix and the block-diagonal matrix, respectively.
$\odot$ denotes the Hadamard product.
The pseudoinverse of $\mathbf{A}$ is denoted by $\mathbf{A}^{\dagger}$, which equals $\left(\mathbf{A}^{\mathrm{H}} \mathbf{A}\right)^{-1} \mathbf{A}^{\mathrm{H}}$.
$\mathbf{O}_{M, N}$ and $\boldsymbol{0}_{M, 1}$ denote the $M \times N$ zero matrix and $M \times 1$ zero vector.
For two matrices $\mathbf{A}$ and $\mathbf{B}$, $[\mathbf{A}, \mathbf{B}]$ denotes the commutator $\mathbf{A}\mathbf{B}-\mathbf{B}\mathbf{A}$, and $\{\mathbf{A}, \mathbf{B}\}$ denotes their anti-commutator $\mathbf{A B}+\mathbf{B A}$.
$a_{0}=52.9~\mathrm{pm}$ denotes the Bohr radius.
$j=\sqrt{-1}$ is an imaginary unit.
$\mathbb{E}\{\cdot\}$ denotes the expectation.
The notation $\mathcal{N}(\boldsymbol{\mu},\sigma^2\mathbf{I})$ denotes a Gaussian distribution whose mean is $\boldsymbol{\mu}$ and covariance is $\sigma^2\mathbf{I}$.
The operators $\Re\{\cdot\}$ and $\Im\{\cdot\}$ respectively denote the real and imaginary parts of a complex number. 
The floor of a real number $x$ is $\lfloor x \rfloor$.

\section{Preliminaries of Rydberg Atomic Quantum Receivers for AFDM Detection and Processing}\label{Sec-II}
In this section, we provide the fundamental principles of the Rydberg atom that describe the interaction with \ac{RF} signals.
Especially, we explain the quantum behavior and system model for the \ac{AFDM} signal reception and processing.

\subsection{Fundamentals of Rydberg Atom}\label{Sec-II-A}

\subsubsection{Quantum state} 

The energy level of an atom changes as the photon is either absorbed or emitted. 
By leveraging this electron transition, the \ac{RF} signals can be detected. 
The electron transition can be modeled by different quantum states of the Rydberg atom. 
These quantum states include the ground state $|1\rangle$, a lowly-excited state $|2\rangle$, and the Rydberg states, which are $|3\rangle$ and $|4\rangle$.
The probe beam of angular frequency $\omega_{p}$ excites the quantum state from $|1\rangle \rightarrow |2\rangle$, and the coupling beam of angular frequency $\omega_{c}$ induces the transition $|2\rangle \rightarrow |3\rangle$, transforming the alkali-metal atom to the Rydberg atom.
Thereafter, the RF signals of angular frequency $\omega_{\mathrm{RF}}$ excite the Rydberg state to another Rydberg state $|3\rangle \rightarrow |4\rangle$, enabling the \ac{RF} signal detection by monitoring the variations induced by these electron transitions via \ac{PD} \cite{cui_MIMO}.

\subsubsection{Rabi frequency and detuning} 

The interaction strength between the \ac{RF} signals and the electric dipole moment is characterized by the \textit{Rabi frequency} \cite{quantum_rotating}. 
The general expression of Rabi frequency $\Omega_{\mathrm{RF}}$ is given by \cite{cui2025rydberg_Mag}:
\begin{equation}
\Omega_{\mathrm{RF}}=\frac{\mu_{34}}{\hbar}|E_{\mathrm{RF}}|,
\label{General_Rabi}
\end{equation}
where $\mu_{34}$, $\hbar$, and $E_{\mathrm{RF}}$ are the transition dipole moment, reduced Planck constant, and RF signals, respectively. 

Furthermore, the frequency deviation between the transition frequency and the carrier frequency arises for every electron transition process due to the discrete energy levels, also referred to as \textit{frequency detuning}. 
For example, the frequency detuning $\Delta_{p}$ presents the gap between $\omega_{p}$ and the transition frequency $\omega_{12}$, i.e., $\Delta_{p}=\omega_{p}-\omega_{12}$.
Likewise, the detuning of the coupling beam $\Delta_{c}$ and the RF signal $\Delta_{\mathrm{RF}}$ can be introduced in the same manner so that $\Delta_{c}=\omega_{c}-\omega_{23}$ and $\Delta_{\mathrm{RF}}=\omega_{\mathrm{RF}}-\omega_{34}$, where $\omega_{23}$ and $\omega_{34}$ are transition frequencies of the coupling beam and RF signals, respectively.
The detunings are zero when they are on-resonant with their respective electron transitions. In this work, we assume the probe beam and the coupling beam are on-resonant with their transitions (i.e., $\Delta_{p}=\Delta_{c}=0$) \cite{cui2025realizing}.

\subsubsection{Dynamics of quantum state}
The dynamic quantum state is governed by the Lindblad master equation.
For the \ac{AFDM}, we adopt the four-level system model since the deviation of each chirp-subcarrier frequency from the transition frequency can be fully described by the different frequency detunings.
Based on this system, the Lindblad master equation is represented as \cite{jing2020atomic}:
\begin{equation}
\frac{\partial \boldsymbol{\rho}}{\partial t}=-\frac{j}{\hbar}[\mathbf{H}, \boldsymbol{\rho}]+\mathcal{L},
\label{Lindblad}
\end{equation}
where $\mathbf{H}$, $\boldsymbol{\rho}$, and $\mathcal{L}$ are the Hamiltonian operator, density matrix, and the decoherence operator, respectively. 

Here, the Hamiltonian operator $\mathbf{H}$ is given by
\begin{equation}
\mathbf{H}=\frac{\hbar}{2}\left[\begin{array}{cccc}
0 & \Omega_{p} & 0 & 0 \\
\Omega_{p} & 0 & \Omega_{c} & 0 \\
0 & \Omega_{c} & 0 & \Omega_{\mathrm{RF}} \\
0 & 0 & \Omega_{\mathrm{RF}} & -2 \Delta_{\mathrm{RF}}
\end{array}\right],
\label{Hamiltonian}
\end{equation}
where $\Omega_{p}$, $\Omega_{c}$, $\Omega_{\mathrm{RF}}$, and $\Delta_{\mathrm{RF}}$ are Rabi frequencies of probe beam, coupling beam, \ac{RF} signal, and the frequency detunings of \ac{RF} signal, respectively. 
The decoherence matrix $\mathcal{L}$ is presented as \cite{jing2020atomic}:
\begin{equation}
\mathcal{L}=-\frac{1}{2}\{\boldsymbol{\Gamma}, \boldsymbol{\rho}\}+\boldsymbol{\Lambda},
\label{Decay_Matrix}
\end{equation}
where $\boldsymbol{\Gamma}=\mathrm{diag}\left\{0, \gamma_2, \gamma_3, \gamma_4\right\}$. 

Here, $\gamma_{j}$ are the decay rates of the $j$-th level.
The decay matrix is presented as $\boldsymbol{\Lambda}=\mathrm{diag}\left\{\gamma_2 \rho_{22}+\gamma_4 \rho_{44}, \gamma_3 \rho_{33}, 0,0\right\}$.

\subsubsection{Optical measurement model}
The optical measurement, which corresponds to the probe beam output power, can be acquired by solving the steady-state solution of $\rho_{12}$, which is the $(1,2)$-th element of the density matrix $\boldsymbol{\rho}$
Here, the $\rho_{12}$ is given by \cite{RydModel_RuiNi2023}:
\begin{equation}
\rho_{12}=\frac{A_{1} \Omega_{\mathrm{RF}}^2 \Delta_{\mathrm{RF}}^2+j B_1 \Omega_{\mathrm{RF}}^4}{C_1 \Omega_{\mathrm{RF}}^4+C_2 \Omega_{\mathrm{RF}}^2+C_3 \Delta_{\mathrm{RF}}^2},
\label{rho_12}
\end{equation}
where $A_1=2 \Omega_p \Omega_c^2$, $B_1=\gamma_2 \Omega_p$, $C_1=2 \Omega_p^2+\gamma_2^2$, $C_{2}=2 \Omega_p^2\left(\Omega_c^2+\Omega_p^2\right)$, and $C_3=4\left(\Omega_c^2+\Omega_p^2\right)^2$. 

Let $P_{\mathrm{in}}$ represents the input power of the probe beam. Then, according to the adiabatic approximation, the output power of the probe beam $P_{\mathrm{out}}$ is defined by the imaginary part of the $(1,2)$-th entry of $\boldsymbol{\rho}$, $\rho_{12}$, which is presented as \cite{jing2020atomic}:
\begin{equation}
P_{\mathrm{out}}=  P_{\text {in}} \exp \left(-C_{0} \operatorname{Im}\left\{\rho_{12}\right\}\right),
\label{Output_Power}
\end{equation}
where the constant $C_{0}$ is given by $C_{0} \triangleq \frac{2 N_0 \mu_{12}^2 k_p L}{\epsilon_0 \hbar \Omega_{p}}$.
Here, $N_{0}$, $\mu_{12}$, ${\epsilon}_{0}$, $k_{p}$, and $L$ are total density of atoms, transition dipole moment of transition $|1\rangle \rightarrow |2\rangle$, vacuum permittivity, length of vapor cell, and wavenumber of probe beam, respectively.
After the readout of the probe beam power, the \ac{PD} converts it into the current $I_{\mathrm{out}}=\frac{q \eta}{\hbar \omega_{\mathrm{p}}} P_{\mathrm{out}}$, where $\eta$ and $q$ are the quantum efficiency of the \ac{PD} and the charge of electrons, respectively.
Then, the output voltage is given by \cite{cui2025realizing}
\begin{equation}
V_{\mathrm {out }}=R_{\mathrm{T}} I_{\mathrm {out }} \triangleq V_{\mathrm {in }} \exp \left(-C_0 \operatorname{Im}\left\{\rho_{12}\right\}\right),
\label{Output_Voltage}
\end{equation}
where $R_{\mathrm{T}}$ is the load impedance and $V_{\mathrm{in}} \triangleq \frac{R_{\mathrm{T}} q \eta}{\hbar \omega_{\mathrm{p}}} P_{\mathrm{in}}$ is the input voltage.

To ease the representation, we introduce the bias function $\Pi(\Omega, \Delta)$ of probe beam output power, obtained by substituting \eqref{rho_12} into \eqref{Output_Power}, which is given by \cite{cui2025realizing}:
\begin{equation}
\Pi(\Omega, \Delta) \triangleq V_{\mathrm{in }} \exp \left\{-\frac{B_1 C_0 \Omega^4}{C_1 \Omega^4+C_2 \Omega^2+C_3 \Delta^2}\right\},
\label{bias}
\end{equation}
where $\Omega \in[0,+\infty)$ and $\Delta \in \mathbb{R}$ are the general notation of the Rabi frequency and detuning, respectively. 

Furthermore, the gain of the probe beam power $\Upsilon(\Omega, \Delta)$ is derived as a partial derivative of $\Pi(\Omega, \Delta)$, given by \cite{cui2025realizing}:
\begin{align}\label{eq:Upsilon}
&\Upsilon(\Omega, \Delta) \triangleq \frac{\partial \Pi(\Omega, \Delta)}{\partial 
\Omega} =  -2V_{\rm in} B_1C_0\times \\&\exp\left\{
-\frac{B_1C_0\Omega^4}{C_1\Omega^4 + C_2\Omega^2 + C_3\Delta^2}
\right\} \frac{\Omega^3(C_2\Omega^2 + 2C_3\Delta^2)}{(C_1\Omega^4 + C_2\Omega^2 + C_3\Delta^2)^2}.\notag
\end{align} 

\subsection{Utilization of Self-Heterodyne Sensing}\label{Sec-II-B}
To extract the channel information from the targets, the \ac{LO} is employed for transmitting the reference signal, which is referred to as heterodyne sensing \cite{jing2020atomic}. 
However, introducing additional reference sources induces a bulky receiver architecture with insufficient instantaneous bandwidth\footnote{Note that the instantaneous bandwidth of \ac{RAQRs} is typically less than 10 MHz \cite{yang2024highly}, which might be insufficient for accurate target sensing in a \ac{DD} channel.}, which may not be suitable for \ac{AFDM} signal detection.
Therefore, we adopt the \textit{self-heterodyne sensing} technique, in which the transmitter itself acts as \ac{LO} \cite{cui2025realizing}.
In this system, the Rabi frequency $\Omega_{\mathrm{RF}}$ in \eqref{General_Rabi} is redefined as the superposition of the Rabi frequency of the reference signal and the target signal, which is 
\begin{equation}
\Omega_{\mathrm{RF}}(t)=\left|\Omega_{l}(t)+\Omega_{s}(t) e^{j \Delta\phi(t,\tau,\tau_{l})}\right|,
\label{Rabi_RF}
\end{equation}
where $\Omega_{l}(t)=\frac{\mu_{34}}{\hbar}|E_{l}(t)|$ and $\Omega_{s}(t)=\frac{\mu_{34}}{\hbar}|E_{s}(t)|$ are the Rabi frequencies of the reference signal $E_{l}(t)$ and target signal $E_{s}(t)$, respectively.
The phase difference between the reference signal and the target signal is presented as
\begin{equation}
\Delta\phi(\tau,\tau_{l})=\phi(t-\tau)-\phi(t-\tau_{l}),
\label{phase_diff}
\end{equation}
where $\tau$ and $\tau_{l}$ are the propagation delay from target and \ac{LO} to the \ac{RAQRs}, respectively.

The phase function $\phi(\cdot)$ analysis of \ac{AFDM} for the self-heterodyne sensing in the \ac{DD} channel will be elaborated in Section~\ref{Sec-III-B}. 
Considering the strong reference approximation, which is $|E_l(t)| \gg |E_s(t)|$, the frequency detuning $\Delta_{\mathrm{RF}}$ is
\begin{equation}
\Delta_{\mathrm{RF}}=\Delta_l=\omega_{{l}}-\omega_{34},
\label{detuning_freq}
\end{equation}
where $\omega_{l}$ denotes the angular frequency of reference signal.

Eventually, by substituting \eqref{Rabi_RF} and \eqref{detuning_freq} into the $V_{\mathrm{out}}(t)=\Pi(\Omega_{\mathrm{RF}}(t), \Delta_{\mathrm{RF}})$, the measured voltage is given by \cite{cui2025realizing}
\begin{align}\label{measured_voltage}
y(t) =
&\Pi(\Omega_{\mathrm{RF}}(t), \Delta_{\mathrm{RF}}) + n(t)\, {\approx}\, \Pi\left(\Omega_{l}(t), \Delta_{l}\right)
\\
&+\frac{\mu_{34}}{\hbar}\Upsilon\left(\Omega_{l}(t), \Delta_{l}\right)  |E_{s}(t)| \cos (\Delta \phi(\tau,\tau_{l}))+n(t) \notag, 
\end{align}
where linearization comes from the first-order Taylor expansion of $\Pi(\Omega_{\mathrm{RF}}(t), \Delta_{\mathrm{RF}})$ \cite{cui2025realizing}.

The noise $n(t)$ follows the Gaussian distribution $\mathcal{N}\left(0, \sigma^2(t)\right)$, where its power $\sigma^2(t)=\sigma_{\mathrm{int}}^2(t)+\sigma_{\mathrm{ext}}^2(t)$ comprises internal sources $\sigma_{\mathrm{int}}^2(t)$ and external sources $\sigma_{\mathrm{ext}}^2(t)$.
The internal noise arises from the randomness in the optical detection process, and the external noise comes from the black-body radiation and quantum fluctuations \cite{tu2024approaching}.
Here, both the power of internal noise and external noise are presented as \cite{cui2025realizing}
\begin{align}\label{noise_intrinsic}
&\sigma^2_{\mathrm{int}}(t)=q R_{\mathrm{T}} \Pi\left(\Omega_l(t), \Delta_l\right), \notag \\ 
&\sigma^2_{\mathrm{ext}}(t)=\frac{\mu_{34}^2}{\hbar^2} \Upsilon^2\left(\Omega_l(t), \Delta_l\right)\left\langle E_I^2\right\rangle, 
\end{align}
where $\left\langle E_I^2\right\rangle=\frac{\hbar \omega_{\mathrm{RF}}^3}{\pi \epsilon_0 c^3}\left(2 n_{\mathrm{th}}+1\right)$ is the field intensity of blackbody radiation.

Here, $\epsilon_0$ is the permittivity of free space, $c$ is the speed of light, and $n_{\mathrm{th}}=1 /\left(e^{\hbar \omega_{\mathrm{RF}} / k_B T_E}-1\right)$ is the Bose-Einstein distribution, where $k_{B}$ and $T_{E}$ are Boltzmann constant and the ambient temperature, respectively.

\begin{figure}[t]
\centering
\includegraphics[width=\linewidth]{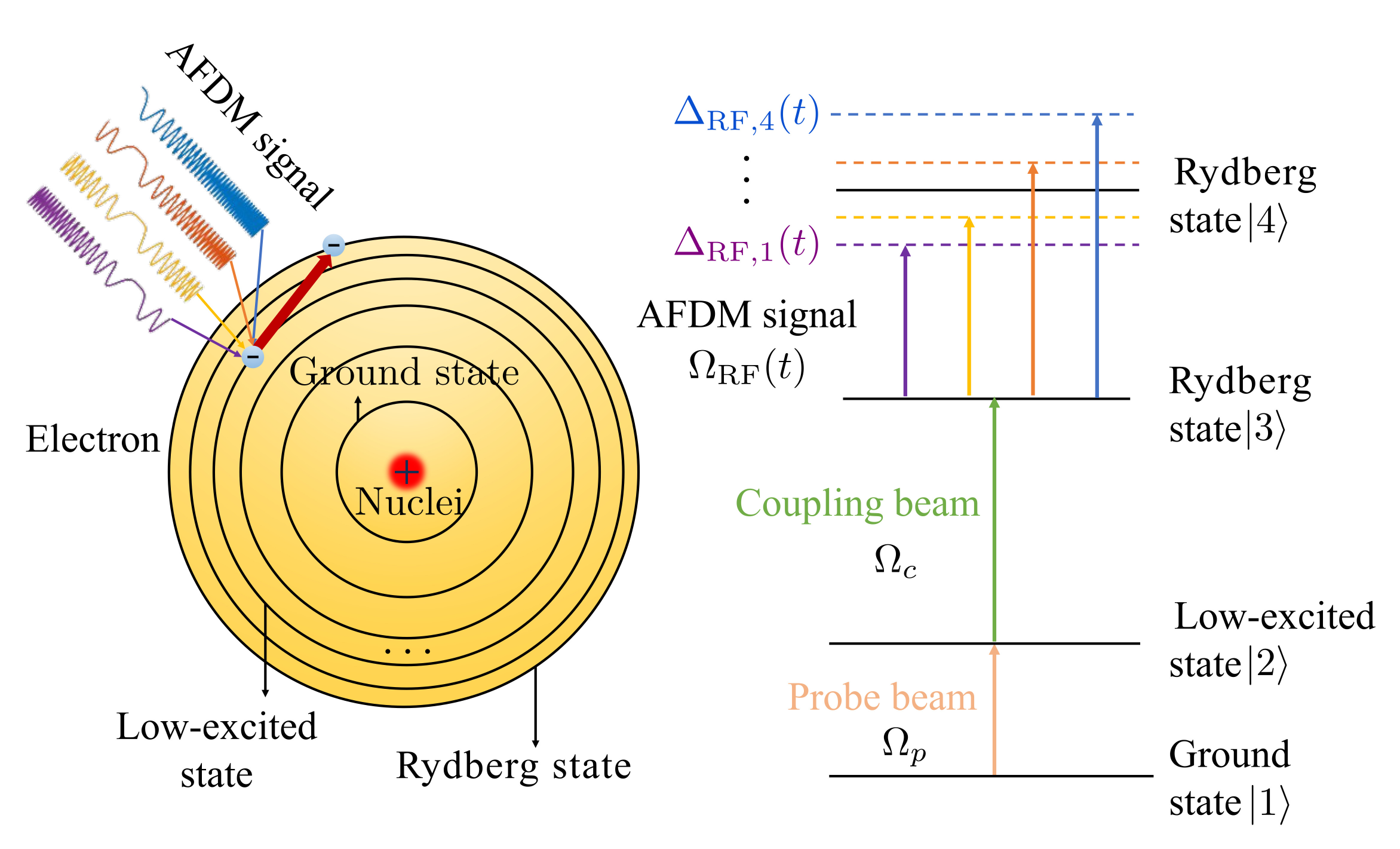}
\caption{Schematic diagram of the electron transitions and four-level quantum system for AFDM detection. Here, the number of chirp-subcarriers $N$ is 4.}
\label{Atom_AFDM}
\vspace{-2ex}
\end{figure}

\subsection{Influence of AFDM Property}\label{Sec-II-C}

Based on the general expression of the measured voltage output in \eqref{measured_voltage}, the measured signal can be modified according to the property of impinging \ac{RF} signals.
Here, we introduce the two important properties of \ac{AFDM} and design a new measured signal model based on these characteristics.
\subsubsection{Multicarrier property} \ac{AFDM} is a multicarrier waveform, where the data is modulated via the \ac{AFT} \cite{AFDM_Bemani_TWC, Rou_SPM24}.
Therefore, unlike the single carrier-based waveform, which only considers a single detuning $\Delta_{\mathrm{RF}}$, the multicarrier of AFDM induces multiple frequency detunings, $\Delta_{\mathrm{RF},m}\triangleq\Delta_{l,m}=\omega_{l,m}-\omega_{34}$ for $m=0,1,\cdots,N-1$, where $\omega_{l,m}$ is the angular frequency of the $m$-th subcarrier, $m$ is the index of subcarriers, and $N$ is the number of subcarriers for \ac{AFDM}, respectively.
Likewise, the phase function, received target signal, and the noise of \ac{AFDM} are modified as $\Delta\phi_{m}(\tau,\tau_{l})$, $E_{s,m}(t)$, and $n_{m}(t)$.
\subsubsection{Chirp property} Due to the chirp property of \ac{AFDM}, the instantaneous frequency of each chirp-subcarrier, $\omega_{l,m}$ for $m=0,1,\cdots,N-1$, is time-varying.
This property yields the time-varying frequency detunings of the \ac{AFDM}, which is presented as $\Delta_{\mathrm{RF},m}(t)\triangleq\Delta_{l,m}(t)=\omega_{l,m}(t)-\omega_{34}$.

The behavior of electron transitions of the Rydberg atom for \ac{AFDM} detection is illustrated in Fig.~\ref{Atom_AFDM}. 
By reflecting these two key properties, the general measured signal in \eqref{measured_voltage} is modified, where the measured signal at the $m$-th chirp-subcarrier is 
\begin{align}\label{AFDM_Measured_Sig}
y_{m}(t) &\triangleq 
\Pi(\Omega_{\mathrm{RF}}(t), \Delta_{\mathrm{RF},m}(t)) + n_{m}(t) \notag
\\
& \, {\approx} \,\,\Pi\left(\Omega_{l}(t), \Delta_{l,m}(t)\right)+\frac{\mu_{34}}{\hbar}\Upsilon\left(\Omega_{l}(t), \Delta_{l,m}(t)\right) \notag \\ & \times |E_{s,m}(t)| \cos (\Delta \phi_{m}(\tau,\tau_{l}))+n_{m}(t),
\end{align}
where $|E_{s,m}(t)|$ is the signal field strength of the $m$-th subcarrier.

The detailed analysis of the measured signal model for \ac{AFDM} \eqref{AFDM_Measured_Sig} with \ac{RAQRs} will be elaborated in the next Section~\ref{Sec-III}.

\section{System Model and Problem Analysis}\label{Sec-III}
In this section, we review the \ac{RAQRs} system model with conventional \ac{SC-AFDM}.
Thereafter, we analyze the optical ambiguity problem of \ac{SC-AFDM} for delay-Doppler estimation with \ac{RAQRs}.

\subsection{Review of Single-Chirp AFDM} \label{Sec-III-A}
The transmitted \ac{DAFT}-domain vector $x[m]$ is mapped onto the discrete time domain signal $s[n]$ using the \ac{IDAFT}, which is presented as \cite{AFDM_Marios}
\begin{equation}
s[n]=\frac{1}{\sqrt{N}} \sum_{m=0}^{N-1} x[m] e^{j 2 \pi\left(c_2 m^2+\frac{1}{N} m n+c_1 n^2\right)},\label{tx_sig_discrete}
\end{equation}
where $c_{1}$ and $c_{2}$ are the post-chirp and pre-chirp parameters of the \ac{IDAFT}, affecting various properties of \ac{AFDM}.

Given the above, the continuous time version of the transmitted signal in \eqref{tx_sig_discrete} can be written as \cite{yin2025ambiguity}
\begin{equation}
s(t)= \sum_{m=0}^{N-1} x[m] e^{j 2 \pi(c_{2}m^2+\phi_{m}(t))}, \quad 0 \leq t < T,
\label{tx_sig}
\end{equation}
where $T=N \Delta t$, with the Nyquist sampling rate $\frac{1}{\Delta t}$, is the duration of the instantaneous phase function of the $m$-th chirp $\phi_{m}(t)$, defined as piece-wise manner as \cite{AFDM_Marios}
\begin{equation} \label{instantaneous_phase}
\phi_m(t)=\tilde{c}_1 t^2+\frac{m}{T}t-\frac{q}{\Delta t}t, \quad t_{m, q} \leq t<t_{m, q+1},
\end{equation}
with $\tilde{c}_{1}=\frac{c_{1}}{(\Delta t)^2}$, and $t_{m,q}$ is the $q$-th spectrum wrapping point of $m$-th chirp subcarrier.

Notice that $c_1$ controls the frequency dispersion of the signal, which can be changed and optimized according to the Doppler characteristics of the \ac{DD} channel~\cite{tek2025novel}.

\subsection{Sensing Scenario and Received \ac{AFDM} Signal Model} \label{Sec-III-B}

\begin{figure}[t]
\centering
\includegraphics[width=\linewidth]{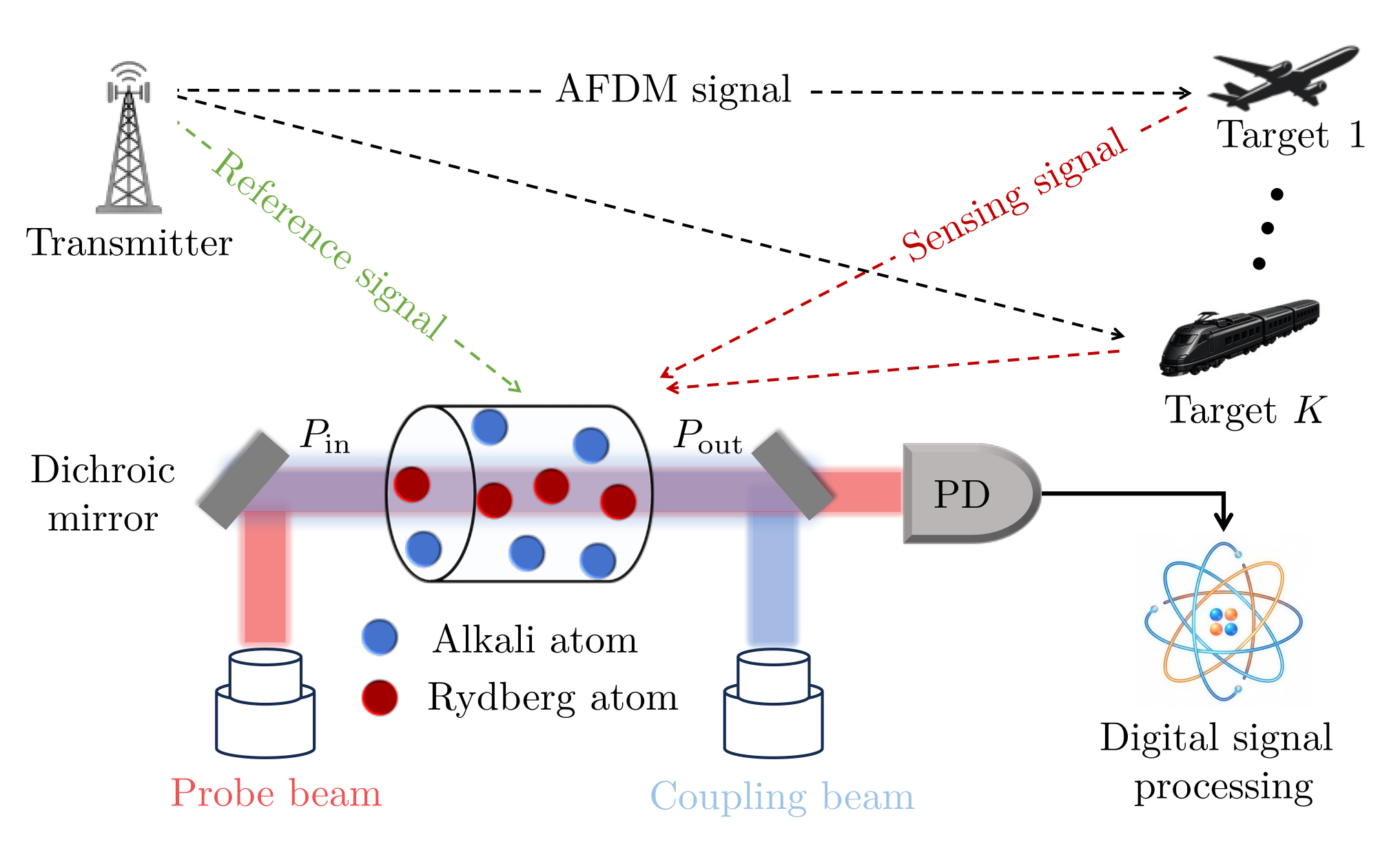}
\caption{Quasi-monostatic multi-target sensing scenario with AFDM signal for delay-Doppler estimation with \ac{RAQRs}.}
\label{Scenario}
\vspace{-2ex}
\end{figure}

Consider a quasi-monostatic sensing of a multi-moving target in the high-mobility scenario, resulting in a \ac{DD} channel~\cite{Rou_SPM24}.
As illustrated in Fig.~\ref{Scenario}, we assume the transmitter broadcasts the \ac{AFDM} signal in downlink to $K$ moving targets and a static \ac{RAQRs}.
Following the received AFDM signal model in \cite{yin2025ambiguity}, the Rabi frequency of the target signal $\Omega_{s}(t)$ and reference signal $\Omega_{l}(t)$ (line-of-sight only) are given by
\begin{equation}\label{Rabi_Freq_Target}
\Omega_{s}(t)=\frac{\mu_{34}}{\hbar}\left|\sum_{k=1}^{K} h_{k}s(t-\tau_{k})e^{j 2 \pi \nu_k t} \right|,
\end{equation}
\begin{equation}\label{Rabi_Freq_Target_2}
\Omega_{l}(t)=\frac{\mu_{34}}{\hbar}\left|h_{l}s(t-\tau_{l})e^{j 2 \pi \nu_l t} \right|,
\end{equation}
where $h_{k},\nu_{k}$ and $\tau_{k}$ are the channel gain, Doppler shift, and the propagation delay of the $k$-th target to \ac{RAQRs}, respectively.

Likewise, $h_{l}$ and $\nu_{l}$ are the channel and Doppler shift of the \ac{LOS} path between transmitter and \ac{RAQRs}. 
Trivially, $\nu_{l}=0$ since the \ac{RAQRs} are fixed without mobility.
In monostatic sensing scenario, the range and velocity of the $k$-th target, $R_{k}$ and $v_{k}$, are given by $R_k=\frac{c\tau_{k}}{2}$ and $v_{k}=\frac{2\pi c\nu_{k}}{2 \omega_{\mathrm{RF}}}$, where $c$ denotes the speed of light.

We employ the self-heterodyne sensing technique \cite{cui2025realizing}, in which the transmitter itself acts as the \ac{LO} to extract the delay-Doppler from the received signal. 
In this scheme, the measured signal's frequency component is determined by the phase difference between the \ac{LO} and the target signal phase profiles.
Based on \eqref{instantaneous_phase}, the phase function for \ac{AFDM} $\Delta\phi_{m}(\tau_{k},\tau_{l},\nu_{k})$ in a \ac{DD} channel is given by
\begin{align} \label{phase_diff_AFDM}
\Delta \phi_m\left(\tau_{k},\tau_{l},\nu_{k}\right) & =\phi_m\left(t-\tau_{k}\right)-\phi_m\left(t-\tau_{l}\right)+\nu_{k}t  \\
& =\left(2 \tilde{c}_1\left(\tau_{l}-\tau_{k}\right)+\nu_{k}\right) t  \nonumber \\ &+\left(\tau_{k}-\tau_{l}\right)\left(-\frac{m}{T}+\frac{q}{\Delta t}+\tilde{c}_1(\tau_{k}+\tau_{l})\right). \nonumber
\end{align}

Furthermore, the instantaneous frequency of \ac{AFDM} coincides with the derivative of $\phi_{m}(t)$ \cite{AFDM_Marios}. 
Therefore, the frequency detuning of the reference signal at the $m$-th chirp subcarrier $\Delta_{l,m}(t)$ is presented as
\begin{equation} \label{detuning_ref}
\begin{aligned}
\Delta_{l,m}(t) & =\frac{d \phi_{m}(t-\tau_{l})}{dt}-\omega_{34} \\ & = 2\tilde{c}_{1}(t-\tau_{l})+\frac{m}{T}-\frac{q}{\Delta t}-\omega_{34}.
\end{aligned}
\end{equation}

The noise of the $m$-th chirp $n_{m}(t)$ is modeled as Gaussian noise $n_{m}(t)$, where its autocorrelation satisfies $\mathrm{E}\left(n_{m}\left(t^{\prime}\right) n_{m}(t)\right)=\sigma_{m}^2(t) \delta\left(t-t^{\prime}\right)$, and is given by the superposition of the intrinsic noise $\sigma_{\mathrm{int},m}^2(t)$ and the extrinsic noise $\sigma_{\mathrm{ext},m}^2(t)$, which are 
\begin{equation}\label{noise_intrinsic_2}
\sigma_{\mathrm{int},m}^2(t)=q R_{\mathrm{T}} \Pi\left(\Omega_l(t), \Delta_{l,m}(t)\right),
\end{equation}
\begin{equation}\label{noise_intrinsic_3}
\sigma_{\mathrm{ext},m}^2(t)=\frac{\mu_{34}^2}{\hbar^2} \Upsilon^2\left(\Omega_l(t), \Delta_{l,m}(t)\right)\left\langle E_I^2\right\rangle. 
\end{equation}
%
Given the above, the noise power $\sigma_{m}^2(t)$ is given by
\begin{equation} \label{noise_tot}
\sigma_{m}^2(t)=\sigma_{\mathrm{int},m}^2(t)+\sigma_{\mathrm{ext},m}^2(t).
\end{equation}
Eventually, the received \ac{AFDM} signal at the $m$-th chirp subcarrier is presented as
\begin{align}\label{power_AFDM}
y_{m}(t)& =  \;\Pi\left(\Omega_l(t), \Delta_{l, m}(t)\right)+\frac{\mu_{34}}{\hbar} \Upsilon\left(\Omega_l(t), \Delta_{l, m}(t)\right) \\
& \times \frac{1}{N}\sum_{k=1}^{K} |E_{s,k,m}(t)| \cos \left(\Delta \phi_m\left(\tau_{k}, \tau_{l}, \nu_k\right)\right)+n_m(t), \notag
\end{align}
where $|E_{s,k,m}(t)|$ is the field strength of the $m$-th subcarrier signal from the $k$-th target. 

Here, we consider the fixed power transmission so that the sum of the field strength $|E_{s,k,m}(t)|$ is defined as $\sum_{k=1}^{K}|E_{s,k,m}(t)|^2=P_{s}$, where $P_{s}$ denotes the received signal power.

\subsection{Joint Delay-Doppler Estimation Problem with \ac{RAQRs}}\label{Sec-III-C}

In light of the measured signal model formulation of \ac{AFDM} in a \ac{DD} channel under \ac{RAQRs}, the objective of this paper is to accurately estimate the delay-Doppler information of the target $[\tau,\nu_{r}]$ from the measurements.
To achieve this goal, we leverage the atomic autocorrelation property of self-heterodyne sensing, which maps the target's delay to a fluctuation frequency of the output voltage of \ac{RAQRs} \cite{cui2025realizing}.
However, the classical \ac{SC-AFDM} brings a critical problem for accurate delay-Doppler estimation with \ac{RAQRs}. 

To begin with, the instantaneous phase function in \eqref{phase_diff_AFDM} is modeled as the fluctuation frequency and the phase of the output voltage, which is presented as
\begin{equation} \label{phase_re}
\Delta \phi_m\left(\tau_{k}, \tau_{l}, \nu_k\right)=-\omega_{m,k}t-\varphi_{m,k},
\end{equation}
where $\omega_{m,k}$ and $\varphi_{m,k}$ are fluctuation frequency and phase of the $m$-th chirp-subcarrier from the $k$-th target.

Then, by substituting \eqref{phase_re} into the \eqref{phase_diff_AFDM}, the fluctuation frequency $\omega_{m,k}$ is given by
\begin{equation} \label{probe_freq}
\omega_{m,k}=2 \tilde{c}_1\left(\tau_{k}-\tau_{l}\right)-\nu_k.
\end{equation}

From \eqref{probe_freq}, it is observed that the two key parameters to be estimated, time delay $\tau_{k}$ and the Doppler shift $\nu_{k}$ of the $k$-th target, are correspondingly embedded in the fluctuation frequency $\omega_{m,k}$.
In other words, this atomic autocorrelation property of self-heterodyne sensing converts the delay-Doppler estimation into the fluctuation frequency estimation problem.

For the $m$-th chirp-subcarrier, we define the fluctuation frequency vector $\boldsymbol{\omega}_{m}$ as
\begin{equation} \label{Measured_Freq_Vec}
\boldsymbol{\omega}_{m}\triangleq\left[\begin{array}{c}
\omega_{m,1} \\
\omega_{m,2} \\
\vdots \\
\omega_{m,K}
\end{array}\right]=\left[\begin{array}{c}
2 \tilde{c}_1 \tau_{1}-2 \tilde{c}_1 \tau_{l}-\nu_1 \\
2 \tilde{c}_1 \tau_{2}-2 \tilde{c}_1 \tau_{l}-\nu_2 \\
\vdots \\
2 \tilde{c}_1 \tau_{K}-2 \tilde{c}_1 \tau_{l}-\nu_K
\end{array}\right].
\end{equation}
\begin{remark}
From the equation \eqref{Measured_Freq_Vec}, we can observe that all the elements of the frequency vector are independent of chirp-subcarrier index $m$, which brings the same frequency vector for all chirp-subcarriers (i.e., $\boldsymbol{\omega}_{0}=\boldsymbol{\omega}_{1}=\cdots=\boldsymbol{\omega}_{N-1}$).
This easily confirms that the estimation problem of $\tau_{k}$ and $\nu_k$ for $k=1,2,\cdots,K$ from the fluctuation frequencies $\{\boldsymbol{\omega}_{0},\boldsymbol{\omega}_{1}\ldots,\boldsymbol{\omega}_{N-1}\}$ is an ill-posed linear inverse problem.
This is because there are a total of $2\times K$ unknown target parameters, $\{\tau_{k}, \nu_{k}\}_{k=1}^{K}$, with only $K$ different equations.
\end{remark}

\section{Multi-Chirp AFDM Waveform Design}\label{Sec-IV}
Throughout our rigorous analysis, we observed that it is challenging to extract the delay-Doppler information with the conventional \ac{SC-AFDM} due to the under-determined estimation model.
To address this problem, we propose the new waveform design of \ac{AFDM} that is compatible with \ac{RAQRs}-based delay-Doppler estimation, termed \ac{MC-AFDM}.
The core idea of \ac{MC-AFDM} is replacing the unique post-chirp $\tilde{c}_{1}$ in conventional \ac{SC-AFDM} by different multiple post-chirps.
The waveform designs for \ac{MC-AFDM} are described as follows. 

\subsection{Multi-Chirp AFDM}\label{Sec-IV-A}
The proposed \ac{MC-AFDM} utilizes multiple segmented time durations, where the \ac{AFDM} frame of every segmented time duration utilizes different post-chirps consecutively. 
To begin with, we define the transmitted signal of a \ac{MC-AFDM}, which is given by
\begin{equation} \label{STMC-AFDM_Frame}
s(t)= \sum_{m=0}^{N-1} x[m] e^{j 2 \pi(c_{2}m^2+\phi^{(p)}_{m}(t))}, \quad T^{(p-1)} \leq t < T^{(p)},
\end{equation}
where $T^{(p)}-T^{(p-1)}$ and $\phi^{(p)}_{m}(t)$ are the time duration and the phase function for the $p$-th post chirp $c^{(p)}_1$, respectively\footnote{The number of chirp-subcarriers $N$ for each segmented-time durations are all equal.}.

Note that $T^{(0)}=0$ and $T^{(P)}=T_\mathrm{tot}$, where each frame has an equal duration $T^{(p)}-T^{(p-1)}=N\Delta t=\frac{T_{\mathrm{tot}}}{P}$ corresponding to an \ac{AFDM} symbol so that $T^{(p)}=pN\Delta t$.
Here, $P$ is the number of post-chirps.
Then, the phase function $\phi^{(p)}_{m}(t)$ is modified as
\begin{equation} \label{Phase-ST-AFDM_Frame}
\phi^{(p)}_m(t)=\tilde{c}^{(p)}_1 t^2+\frac{m}{N\Delta t}t-\frac{q}{\Delta t}t, 
\end{equation}
where $\tilde{c}^{(p)}_1=\frac{c^{(p)}_{1}}{\Delta t^2}$ is the $p$-th chirp rate with post-chirp $c^{(p)}_{1}$.

Then, by substituting \eqref{Phase-ST-AFDM_Frame} into \eqref{phase_diff_AFDM}, the fluctuation frequency of the $m$-th subcarrier in \eqref{probe_freq} is reformulated as
\begin{equation} 
\omega^{(p)}_{m,k}= 2\tilde{c}^{(p)}_{1}(\tau_{k}-\tau_{l})-\nu_{k}, \quad T^{(p-1)} \leq t < T^{(p)}.
\label{DC-AFDM_Re} 
\end{equation}
%

%

\begin{figure}[t]
\centering
\includegraphics[width=\linewidth]{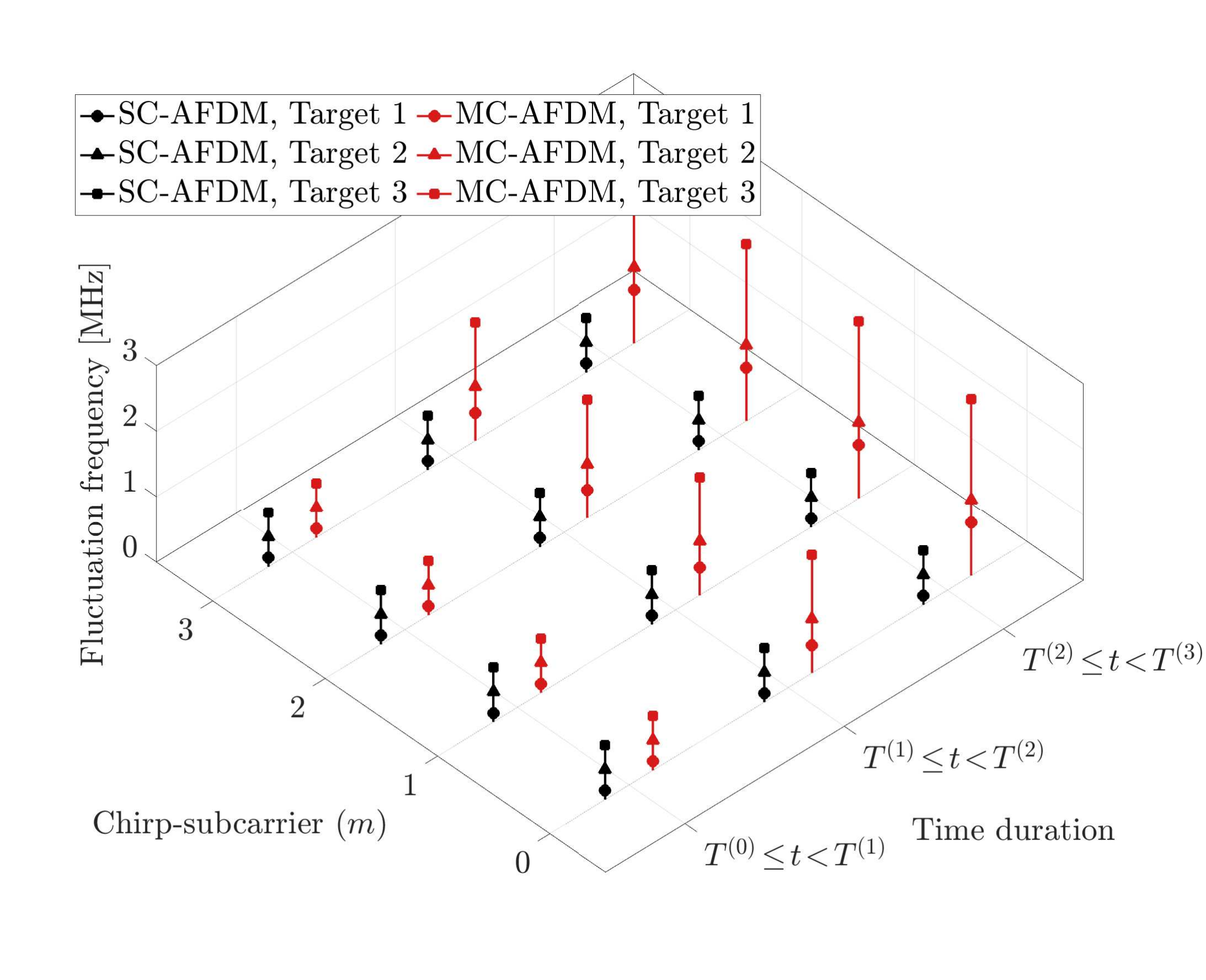}
\caption{3D-fluctuation frequency spectrum analysis of \ac{SC-AFDM} and the proposed \ac{MC-AFDM}. Here, the post-chirp for \ac{SC-AFDM} is set to $c^{(1)}_{1}$, and the post-chirps of \ac{MC-AFDM} are set to $c_{1}^{(1)}$, $2c_{1}^{(1)}$, and $3c_{1}^{(1)}$, respectively.}
\label{Fluct-Spectrum}
\vspace{-2ex}
\end{figure}

From equation \eqref{DC-AFDM_Re}, it is trivial that the problem is now solvable since the utilization of the multiple post-chirps, $c^{(1)}_{1},c^{(2)}_{1},\cdots,c^{(P)}_{1}$, brings the additional different fluctuation frequencies for each time duration.
Specifically, $K$ different fluctuation frequencies in the \ac{SC-AFDM} are increased to the $K \times P$ different fluctuation frequencies in \ac{MC-AFDM}, which transforms the estimation problem from an under-determined system to an over-determined system, thereby enabling the joint delay-Doppler estimation.

For the intuitive understanding, the fluctuation frequency spectrum of the proposed \ac{MC-AFDM} and the conventional \ac{SC-AFDM} is visualized in Fig.~\ref{Fluct-Spectrum}.
Here, the number of targets, the number of chirp-subcarriers, and the number of post-chirps are set to $K=3$, $N=4$, and $P=3$, respectively.
The key observation here is that the conventional \ac{SC-AFDM} shows the same set of target frequency measurements for all time durations, due to the unique and single post-chirp parameter $c^{(1)}_{1}$.
Meanwhile, the target fluctuation frequency measurements of the proposed \ac{MC-AFDM} are distinct for each frame due to the different post-chirps $c^{(1)}_{1}, 2c^{(1)}_{1}$, and $3c^{(1)}_{1}$.
This diversity of optical measurement enables the delay-Doppler estimation with \ac{RAQRs} by solving the ambiguity of the unique optical measurement problem in Remark 1.

\subsection{Multi-Chirp AFDM Measured Signal Model}\label{Sec-IV-B}
The measured model of \ac{AFDM} derived in Sec.~\ref{Sec-III-B} is reformulated to incorporate the $p$-th post-chirp rate $\tilde{c}^{(p)}_{1}$ of the \ac{MC-AFDM}.
By substituting $\tilde{c}^{(p)}_{1}$ into \eqref{detuning_ref}, the frequency detuning of the reference signal at the $m$-th chirp-subcarrier during the $p$-th frame is given by
\begin{align}\label{detuning_ref_MCAFDM}
\Delta_{l,m}^{(p)}(t) & =\frac{d \phi^{(p)}_{m}(t-\tau_{l})}{dt}-\omega_{34} \\ & = 2\tilde{c}^{(p)}_{1}(t-\tau_{l})+\frac{m}{N \Delta t}-\frac{q}{\Delta t}-\omega_{34}. \notag
\end{align} 
Accordingly, the intrinsic and extrinsic noise powers of noise $n^{(p)}_{m}(t)$ are presented as\footnote{Here, the Rabi frequency of the \ac{LO}, $\Omega_{l}(t)$, is independent of $p$ since the \ac{AFDM} signal has a constant envelope for $p=1,2,\cdots,P$, which approximates $\Omega_l(t) \approx \frac{\mu_{34}}{\hbar}\left|h_l\right| \sqrt{P_s}$.}
\begin{equation}\label{noise_reform}
(\sigma^{(p)}_{\mathrm{int},m}(t))^2=q R_{\mathrm{T}} \Pi\left(\Omega_l(t), \Delta^{(p)}_{l,m}(t)\right),
\end{equation}
\begin{equation}\label{noise_reform_2}
(\sigma_{\mathrm{ext},m}^{(p)}(t))^2=\frac{\mu_{34}^2}{\hbar^2} \Upsilon^2\left(\Omega_l(t), \Delta^{(p)}_{l,m}(t)\right)\left\langle E_I^2\right\rangle, 
\end{equation}
%
and the total noise power is given by $(\sigma^{(p)}_{m}(t))^2=(\sigma^{(p)}_{\mathrm{int},m}(t))^2+(\sigma^{(p)}_{\mathrm{ext},m}(t))^2$.

Eventually, the measured signal for \ac{MC-AFDM} at the $m$-th chirp-subcarrier during the $p$-th frame is presented as
\begin{align}\label{power_MCAFDM}
y^{(p)}_{m}(t)& =  \;\Pi(\Omega_l(t), \Delta^{(p)}_{l, m}(t))+\frac{\mu_{34}}{\hbar} \Upsilon(\Omega_l(t),  \Delta^{(p)}_{l, m}(t))  \\
& \times \frac{1}{N}\sum_{k=1}^{K} |E_{s,k,m}(t)| \cos (\Delta \phi^{(p)}_m(\tau_{k}, \tau_{l}, \nu_k))+n^{(p)}_m(t), \notag
\end{align}
where the phase function $\Delta \phi^{(p)}_m(\tau_{k}, \tau_{l}, \nu_k)$ is obtained via substituting \eqref{Phase-ST-AFDM_Frame} into \eqref{phase_diff_AFDM}.

\subsection{Chirp Parameter Design}\label{Sec-IV-C}
We now reveal that the multiple post-chirps enable the \ac{AFDM}-based joint delay-Doppler estimation with \ac{RAQRs}.
However, several important conditions should be strictly followed for compatibility with \ac{RAQRs}, which motivates us to design the post-chirp parameter for \ac{MC-AFDM}.
\subsubsection{Number of post-chirps} Adopting the fluctuation frequency model in \eqref{DC-AFDM_Re}, the fluctuation frequency vector $\boldsymbol{\omega}_{k}$ can be presented as
\begin{equation} \label{Fluctuation frequency}
\underbrace{\left[\begin{array}{c}
\omega_{k}^{(1)} \\
\omega_{k}^{(2)} \\
\vdots \\
\omega_{k}^{(P)}
\end{array}\right]}_{\boldsymbol{\omega}_{k} \in \mathbb{R}^{P \times 1}}=\underbrace{\left[\begin{array}{cc}
2 \tilde{c}_1^{(1)} & -1 \\
2 \tilde{c}_1^{(2)} & -1 \\
\vdots& \vdots \\
2 \tilde{c}_1^{(P)} & -1
\end{array}\right]}_{\triangleq \mathbf{C}_1 \in \mathbb{R}^{P \times 2}} \underbrace{\left[\begin{array}{c}
\tau_{k}-\tau_{l} \\
\nu_k
\end{array}\right]}_{\boldsymbol{\vartheta}_{k}\in \mathbb{R}^{2 \times 1}},
\end{equation}
where $\omega^{(p)}_{k}=\frac{1}{N}\sum_{m=0}^{N-1}\omega^{(p)}_{m,k}$ is the average fluctuation frequency of all chirp-subcarriers at the $k$-th target\footnote{Although $\omega^{(p)}_{m,k}$ is independent of the chirp-subcarrier index $m$ as shown in \eqref{DC-AFDM_Re}, the estimated fluctuation frequency $\hat{\omega}^{(p)}_{m,k}$ might be different across chirp-subcarriers due to the error of the fluctuation frequency estimation. Therefore, we define $\omega^{(p)}_{k}$ as the average of the fluctuation frequency of all chirp-subcarriers.}, $\mathbf{C}_1$ is the post-chirp matrix, and $\boldsymbol{\vartheta}_{k}=[\tau_{k}-\tau_{l},\nu_{k}]^{\mathsf{T}}$ is the channel parameter vector of the $k$-th target.
\begin{remark} 
     The post-chirp matrix $\mathbf{C}_{1}$ in \eqref{Fluctuation frequency} is a full-rank if there exist at least two distinct chirp rates within the chirp-rate vector $\mathbf{\tilde{c}}_{1}=[\tilde{c}^{(1)}_{1},\tilde{c}^{(2)}_{1},\cdots,\tilde{c}^{(P)}_{1}]$, which is $\tilde{c}_1^{(p_{1})} \neq \tilde{c}_1^{\left(p_{2}\right)}$ for at least one pair of $p_{1}$ and $p_{2}$.
     This observation confirms that the minimum number of post-chirps required for unambiguous joint delay-Doppler estimation is $P=2$, which coincides with the dual-chirp \ac{AFDM} in our previous work \cite{kim2026dual}.
\end{remark}

\subsubsection{Maximum chirp-rate} Recall that the \ac{RAQRs} have low supportable instantaneous bandwidth.
Due to this characteristic, the maximum chirp-rate, $\tilde{c}_{1}^{(\mathrm{max})}=\max \{\tilde{c}^{(p)}_{1} \mid p \in\{1, \ldots, P\}\}$, should be carefully designed.
The maximum instantaneous bandwidth of the measured signal in \eqref{power_AFDM} should hold the following inequality
\begin{equation} \label{instantaneous_bandwidth}
\frac{|\omega_{m, k}^{(\mathrm{max})}|}{2 \pi}=\frac{|2 \tilde{c}_1^{(\mathrm{max})}\left(\tau_k-\tau_l\right)-\nu_k|}{2\pi} \leq B_{\mathrm{RAQR}},
\end{equation}
where $B_{\mathrm {RAQR}}$ denotes the allowable instantaneous bandwidth of \ac{RAQRs} which is typically lower than 10 MHz.

Since the parameters $[\tau_{k},\nu_{k}]$ are unknown, we consider a range of delay and Doppler, which are $\tau_k \in\left[\tau_{\mathrm{min}}, \tau_{\mathrm{max}}\right]$ and $\nu_k \in\left[-\nu_{\mathrm{max} }, \nu_{\mathrm{max} }\right]$, where $\tau_{\mathrm{max}}$ and $\nu_{\mathrm{max}}$ are maximum delay and Doppler within the range, respectively.
Then, by substituting the largest fluctuation frequency $\omega_{\mathrm{max} }^{(p)}=2 \tilde{c}_1^{(p)} (\tau_{\mathrm{max}}-\tau_{l})+\nu_{\mathrm{max}}$, which assumes $\tau_{k}=\tau_{\mathrm{max}}$ and $\nu_{k}=-\nu_{\mathrm{max}}$, into the equation \eqref{instantaneous_bandwidth}, the condition for maximum chirp-rate is given by
\begin{equation} \label{condition_chirprate}
\tilde{c}_1^{(\mathrm{max})} \leq \frac{2 \pi B_{\mathrm {RAQR}}-\nu_{\max }}{2 (\tau_{\mathrm{max}}-\tau_{l})}.
\end{equation}

Since the \ac{RF} signal bandwidth that is larger than $B_{\mathrm{RAQR}}$ severely damages the sensitivity of \ac{RAQRs} due to the limited response time of the Rydberg atomic quantum system~\cite{yang2024highly,bohaichuk2022origins}, the above condition in \eqref{condition_chirprate} should be strictly followed for \ac{MC-AFDM} waveform design.

\subsubsection{Post-chirp matrix design} Given the maximum chirp-rate condition, the optimization of post-chirp matrix $\mathbf{C}_{1}$ can further enhance the delay-Doppler estimation performance. 
The optimization of $\mathbf{C}_{1}$ can be achieved by determining its well-conditioning, where a poor conditioning of $\mathbf{C}_{1}$ amplifies the fluctuation frequency estimation error.
To be specific, the fluctuation frequency model in \eqref{Fluctuation frequency} is given by
\begin{equation} \label{fluctuation_model}
\boldsymbol{\hat{\omega}}_{k} \triangleq \mathbf{C}_{1}\boldsymbol{\vartheta}_{k}+\boldsymbol{\epsilon},
\end{equation}
where $\boldsymbol{\hat{\omega}}_{k}$ is the estimated fluctuation frequency and $\boldsymbol{\epsilon}$ is a fluctuation frequency estimation error that arises in the frequency estimation procedure.

Then, the delay-Doppler estimation accuracy is governed by a condition number of $\mathbf{C}_{1}$, which is presented as
\begin{equation} \label{condition_number}
\kappa\left(\mathbf{C}_1\right)=\sqrt{\frac{\lambda_1}{\lambda_2}},
\end{equation}
where $\lambda_{1}$ and $\lambda_{2}$ are eigenvalues $(\lambda_{1} \geq \lambda_{2} > 0)$ of the matrix $\mathbf{C}_1^{\mathsf{T}} \mathbf{C}_1$ \cite{edelman1988eigenvalues}.
Here, the maximizing delay-Doppler estimation accuracy with the estimated fluctuation frequency model \eqref{fluctuation_model} is transformed into optimizing the condition number of the post-chirp matrix, which is formulated as
\begin{subequations}
\begin{align}
\text{(P)} \quad & \min _{\{\tilde{c}^{(p)}_{1}\}^{P}_{p=1}}  \kappa\left(\mathbf{C}_1\right),  \quad \forall p \neq p^{\prime},\label{opt_prob}  \\
& \  \text{s.t. }   \tilde{c}_{1}^{(\mathrm{min}) } \leq \tilde{c}_1^{(p)} \leq \tilde{c}_{1}^{(\mathrm{max}) }, \ \ \tilde{c}_1^{(p)} \neq \tilde{c}_1^{\left(p^{\prime}\right)} .
\label{opt_const} 
\end{align}
\end{subequations}
where $ \tilde{c}_{1}^{(\mathrm{min}) }$ denotes the minimum post chirp-rate.
\begin{theorem} \label{Theorem_1} 
     For a fixed number of $P$, the condition number of the post-chirp matrix $\kappa\left(\mathbf{C}_1\right)$ is minimized when the variance of post-chirp rate vector $\mathrm{Var}(\mathbf{\tilde{c}}_{1})$ is maximized.
     The maximization of $\mathrm{Var}(\mathbf{\tilde{c}}_{1})$ can be achieved by concentrating the chirp rates toward the edge of two values $\tilde{c}_{1}^{(\mathrm{min})}$ and $\tilde{c}_{1}^{(\mathrm{max})}$, which is presented as
     %
\begin{equation} 
\tilde{c}_{1}^{(p)}\triangleq \begin{cases}\tilde{c}_{1}^{(\mathrm{min})}+p\Delta \tilde{c}^{(\mathrm{min})}_{1}, & p=1,2, \dotsc,\lfloor \frac{P}{2} \rfloor \\ \tilde{c}_{1}^{(\mathrm{max})}-(P-p) \Delta \tilde{c}^{(\mathrm{min})}_{1}, & p=\lfloor \frac{P}{2} \rfloor+1,\dotsc,P \end{cases}
    \label{chirp_theorem_1}
\end{equation}
where $\Delta \tilde{c}^{(\mathrm{min})}_{1}$ denotes the minimum chirp-rate difference.
\end{theorem}
\begin{proof}
    (See Appendix~\ref{app2}). 
\end{proof}

\begin{figure}[t]
\centering
\includegraphics[width=1\linewidth]{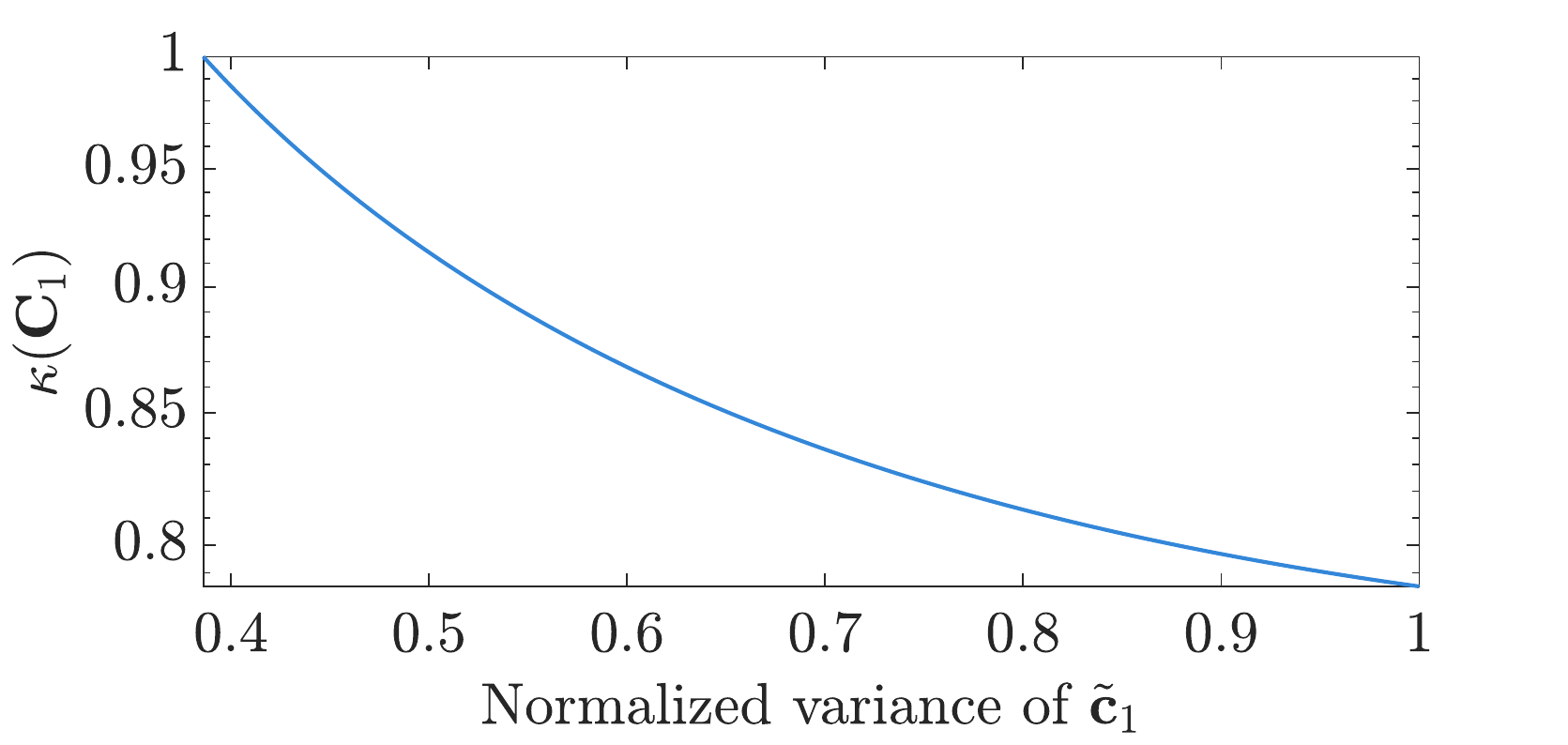}
\caption{Condition number of the post-chirp matrix $\kappa\left(\mathbf{C}_1\right)$ according to the normalized variance of $\tilde{\mathbf{c}}_{1}$. Here, the variance is set by varying the minimum post-chirp difference $\Delta \tilde{c}^{(\mathrm{min})}_{1}$.}
\label{kapaa_var}
\vspace{-2ex}
\end{figure}
The detailed analysis on $\tilde{c}_{1}^{(\mathrm{min})}$ and $\Delta \tilde{c}^{(\mathrm{min})}_{1}$ is elaborated in Remark~\ref{min_chirp}.
To intuitively understand the Theorem~\ref{Theorem_1}, Fig.~\ref{kapaa_var} shows the condition number of the post-chirp matrix $\kappa\left(\mathbf{C}_1\right)$ versus the normalized variance of $\mathbf{\tilde{c}}_{1}$.
We can observe that the condition number is decreased exponentially with respect to the variance of the chirp-rate vector, which supports Theorem~\ref{Theorem_1}.
The reduction of the condition number is beneficial for delay-Doppler estimation with \ac{MC-AFDM}, where its impact is analyzed in Sec.~\ref {Sec-VI}.

\section{Joint Delay-Doppler Estimation Algorithm with Multi-Chirp AFDM}\label{Sec-V}
In this section, we propose a joint delay-Doppler estimation algorithm with the \ac{MC-AFDM} signal.
The proposed algorithm comprises two sequential stages, which are fluctuation frequency estimation and \ac{LS}-based delay-Doppler estimation.
Furthermore, we analyze the theoretical lower bound of the delay and Doppler estimation for \ac{MC-AFDM}.

\begin{algorithm}[t]\label{algorithm1}
  \caption{\ac{OMP}-based delay-Doppler estimation algorithm for \ac{MC-AFDM} with \ac{RAQRs}}
  \Input{$\bar{y}^{(p)}_{m}(t)$, $\Omega_{l}(t)$, $\Delta^{(p)}_{l,m}(t)$, $N$, $P$, $K$, $G_{f}$, $S_{p}$}

  \For{$p={1,2,\cdots,P}$}{

  Calculate $\mathbf{y}_m^{(p)}$
  
  Construct dictionary matrix $\boldsymbol{\Psi}^{(p)}$

  \For{$m={0,1,\cdots,N-1}$}{

  Initialize $\mathbf{r}_{m}^{(p,0)}=\mathbf{y}_m^{(p)}$, $\mathcal{S}_{m}^{(p,0)}=\emptyset$, and $\boldsymbol{\Phi}_m^{(p, 0)}=\emptyset$

  \For{$l={1,2,\cdots,K}$}{ 
  $g^{\star}=\underset{g\in\{1,2,\cdots G_{f}\}}{\arg \max }\left|\left(\psi_g^{(p)}\right)^\mathsf{H} \mathbf{r}_{m}^{(p,l-1)}\right|$
  
   $\mathcal{S}_{m}^{(p,l)}=\mathcal{S}_{m}^{(p,l-1)} \cup g^{\star}$

   $\boldsymbol{\Phi}_m^{(p, l)}=[\boldsymbol{\Phi}_m^{(p, l-1)}, \psi_{g^{{\star}}}^{(p)}]$

   $\hat{\mathbf{c}}_{m}^{(p,l)}=(\boldsymbol{\Phi}_m^{(p, l)})^{\dagger} \mathbf{y}_m^{(p)}$

   $\mathbf{r}_{m}^{(p,l)}=\mathbf{y}_m^{(p)}-\boldsymbol{\Phi}_m^{(p, l)} \hat{\mathbf{c}}_{m}^{(p,l)}$

  }

  $\hat{\omega}_{m, k}^{(p)}=\omega_{g^{(\star,k)}}$

  }

    $\hat{\omega}_k^{(p)}=\frac{1}{N} \sum_{m=0}^{N-1} \hat{\omega}_{m, k}^{(p)}$
  
  }
  \For{$k=1,2,\cdots,K$}{
  Construct estimated frequency vector $\hat{\boldsymbol{\omega}}_k$

  Calculate $\boldsymbol{\vartheta}_{k}=\left(\mathbf{C}_1^{\top} \mathbf{C}_1\right)^{-1} \mathbf{C}_1^{\top} \hat{\boldsymbol{\omega}}_{k}$

  Estimate $\hat{\tau}_{k}=[\boldsymbol{\vartheta}_{k}]_{1}+\tau_{l}$ and $\hat{\nu}_{k}=[\boldsymbol{\vartheta}_{k}]_{2}$
  }

\Output{Estimated delay and Doppler $\left\{\hat{\tau}_k,\hat{\nu}_k\right\}_{k=1}^K$}
\end{algorithm}

\subsection{Fluctuation Frequency Estimation}\label{Sec-V-A}

To estimate the fluctuation frequency from the measured signal in \eqref{power_AFDM}, we adopt the \ac{NLS}-based optimization method in \cite{cui2025realizing}.
First, we define the normalized received signal of the \ac{MC-AFDM} at the $m$-th chirp-subcarrier, which is given by
\begin{align} \label{norm_received}
\bar{y}^{(p)}_{m}(t) & =\frac{y^{(p)}_{m}(t)-\Pi\left(\Omega_l(t), \Delta^{(p)}_{l,m}(t)\right)} {\sigma^{(p)}_{m}(t)} \notag \\
& =\sum^{K}_{k=1}\varrho^{(p)}_{m,k}(t) h_{l} \cos (\omega^{(p)}_{m,k} t+\varphi^{(p)}_{m,k})+\bar{n}^{(p)}_{m}(t),
\end{align}
where $\varphi^{(p)}_{m,k}$ is a phase of \ac{MC-AFDM} and $\bar{n}^{(p)}_{m}(t)\sim \mathcal{N}(0,1)$ is a unit-power Gaussian noise.

Here, the time-varying amplitude $\varrho_{m, k}^{(p)}(t)$ is given by
\begin{equation} \label{norm_gain}
\varrho^{(p)}_{m,k}(t)=\frac{\mu_{34} \Upsilon\left(\Omega_{l}, \Delta^{(p)}_{l,m}(t)\right) \sqrt{P_{s}}}{\hbar\sigma^{(p)}_{m}(t) \sqrt{KN}} .
\end{equation}

Then, following the \ac{NLS} problem in \cite{cui2025realizing}, the multi-target frequency estimation problem for the $p$-th frame is given by
\begin{equation} \label{NLS_Prob}
\max _{{\omega}^{(p)}_{m,k}}\left|\int_{T^{(p-1)}}^{T^{(p)}} \bar{y}^{(p)}_{m}(t) \varrho^{(p)}_{m}(t) e^{j \omega^{(p)}_{m,k} t} \mathrm{~d} t\right|^2,
\end{equation}
for $k=1,2,\cdots,K$. Note that the entire fluctuation frequency of \ac{MC-AFDM} can be acquired by conducting this top-$K$ frequency peak estimation for $p=1,2,\cdots, P$.  

To solve the problem \eqref{NLS_Prob}, we adopt the \ac{OMP} algorithm~\cite{4385788}.
First, let $\mathbf{y}_m^{(p)}\in \mathbb{C}^{S_p \times 1}$ denotes the time sample vector, which is defined as
\begin{equation} \label{ts_vector}
\mathbf{y}_m^{(p)}=\bar{\mathbf{y}}_m^{(p)} \odot \boldsymbol{\varrho}_m^{(p)},
\end{equation}
where $\bar{\mathbf{y}}_m^{(p)}=[\bar{y}^{(p)}_{m}(t_{1}),\bar{y}^{(p)}_{m}(t_{2}),\cdots,\bar{y}^{(p)}_{m}(t_{S_{p}})]^{\mathsf{T}}$ is a normalized received signal vector and $\boldsymbol{\varrho}_m^{(p)}=[\varrho^{(p)}_{m}(t_{1}),\varrho^{(p)}_{m}(t_{2}),\cdots,\varrho^{(p)}_{m}(t_{S_{p}})]^{\mathsf{T}}$ is the average amplitude vector, where $\varrho^{(p)}_{m}(t)=\frac{1}{K}\sum^{K}_{k=1}\varrho^{(p)}_{m,k}(t)$.
Note that $S_{p}=N\Delta t/G_{T}$ is the number of time samples in the $p$-th frame, where $G_{T}$ is the sampling interval.
Then, the dictionary matrix $\boldsymbol{\Psi}^{(p)} \in \mathbb{C}^{S_{p} \times G_{f}}$ is given by
\begin{equation} \label{dictionary}
\boldsymbol{\Psi}^{(p)}=\frac{1}{\sqrt{S_p}}\left[\boldsymbol{\psi}^{(p)}_{1}, \boldsymbol{\psi}^{(p)}_{2}, \cdots, \boldsymbol{\psi}^{(p)}_{G_{f}}\right],
\end{equation}
where $\boldsymbol{\psi}^{(p)}_{g}=[e^{j\omega_{g}t_{1}},e^{j\omega_{g}t_{2}},\cdots,e^{j\omega_{g}t_{S_{p}}}]^{\mathsf{T}} \in \mathbb{C}^{S_{p} \times 1}$, and $G_{f}$ denotes the number of frequency grid points with $\omega_g=\frac{2 \pi B_{\mathrm{RAQR}}}{G_{f}}(g-1)$ for $g=1,2,\cdots,G_{f}$.

Eventually, the sparse presentation of the measurement model is rewritten as
\begin{equation} \label{model_OMP}
\mathbf{y}_m^{(p)} \approx \boldsymbol{\Psi}^{(p)} \mathbf{c}_m^{(p)},
\end{equation}
where $\mathbf{c}_m^{(p)} \in \mathbb{C}^{G_{f} \times 1}$ is a $K$-sparse vector whose non-zero entries correspond to the $K$ fluctuation frequencies.

Based on the sparse signal model in \eqref{model_OMP}, we initialize the residual vector $\mathbf{r}_{m}^{(p,0)}=\mathbf{y}_m^{(p)}$ and support set $\mathcal{S}_m^{(p, 0)}=\emptyset$.
Then, we calculate the correlation and detect the new support, $g^{\star}$, which is presented as
\begin{equation} \label{new_support}
g^{\star}=\underset{g\in\{1,2,\cdots G_{f}\}}{\arg \max }\left|\left(\psi_g^{(p)}\right)^\mathsf{H} \mathbf{r}_{m}^{(p,l-1)}\right|,
\end{equation}
where the support set is updated as $\mathcal{S}_{m}^{(p,l)}=\mathcal{S}_{m}^{(p,l-1)} \cup g^{\star}$.

Then, the vector $\hat{\mathbf{c}}_{m}^{(p,l)}$ is calculated through orthogonal projection $\hat{\mathbf{c}}_{m}^{(p,l)}=(\boldsymbol{\Phi}_m^{(p, l)})^{\dagger} \mathbf{y}_m^{(p)}$, where $(\boldsymbol{\Phi}_m^{(p, l)})$ is updated by $\boldsymbol{\Phi}_m^{(p, l)}=[\boldsymbol{\Phi}_m^{(p, l-1)}, \psi_{g^{{\star}}}^{(p)}]$.
After the projection, the residual vector is updated as $\mathbf{r}_{m}^{(p,l)}=\mathbf{y}_m^{(p)}-\boldsymbol{\Phi}_m^{(p, l)} \hat{\mathbf{c}}_{m}^{(p,l)}$.
The above steps are carried out $K$ times, and the estimated fluctuation frequency is given by $\hat{\omega}_{m, k}^{(p)}=\omega_{g^{(\star,k)}}$, where $g^{(\star,k)}$ is the $k$-th element of $\mathcal{S}^{(p,K)}_{m}$. 

Eventually, the final estimated frequency can be acquired by taking the average across $N$ chirp-subcarriers, which is $\hat{\omega}_k^{(p)}=\frac{1}{N} \sum_{m=0}^{N-1} \hat{\omega}_{m, k}^{(p)}$.
\begin{remark} \label{min_chirp} 
     Since the \ac{OMP} is a discretized grid-based algorithm, the fluctuation-frequency measurements generated by any two distinct post-chirps should remain resolvable on the \ac{OMP} grid for all targets within the delay range $\tau_k \in\left[\tau_{\min}, \tau_{\max}\right]$.
     This condition is presented as $2 \Delta \tilde{c}_{1}(\tau_k-\tau_{l}) \geq \delta_\omega$, where $\delta_\omega=\frac{2 \pi B_{\mathrm{RAQR}}}{G_{f}}$ is the gap between adjacent grid points.
     Considering the smallest time delay difference between \ac{LO} and target (i.e., $\tau_{k}=\tau_{\mathrm{min}}$), the minimum chirp-rate difference $\Delta \tilde{c}^{(\mathrm{min})}_{1}$ is given by
\begin{equation} 
\Delta \tilde{c}_1 \geq \frac{\pi B_{\mathrm{RAQR}}}{G_f\left(\tau_{\min }-\tau_l\right)} \triangleq \Delta \tilde{c}_1^{(\min )}.
    \label{min_chirp_diff}
\end{equation}
\end{remark}

\subsection{Joint Delay-Doppler Estimation}\label{Sec-V-B}
With the estimated fluctuation frequency, the joint delay-Doppler estimation problem can be formulated based on the \ac{LS} criterion.
Following the fluctuation frequency model in \eqref{Fluctuation frequency}, the estimated fluctuation frequency vector $\hat{\boldsymbol{\omega}}_{k} \in \mathbb{R}^{P \times 1}$ is presented as 
\begin{equation} \label{est_fluc_freq}
\hat{\boldsymbol{\omega}}_k=\mathbf{C}_1 \boldsymbol{\vartheta}_k,
\end{equation}
where $\hat{\boldsymbol{\omega}}_k=[\hat{\omega}^{(1)}_{k},\hat{\omega}^{(2)}_{k},\cdots,\hat{\omega}^{(P)}_{k}]^{\mathsf{T}}$.
Then, by applying the \ac{LS} method on \eqref{est_fluc_freq}, we can obtain the parameter vector $\boldsymbol{\vartheta}_k$ as 
\begin{equation} \label{parameter_vector}
\boldsymbol{\vartheta}_{k}=\left(\mathbf{C}_1^{\top} \mathbf{C}_1\right)^{-1} \mathbf{C}_1^{\top} \hat{\boldsymbol{\omega}}_{k},
\end{equation}
which yields the estimated delay and Doppler as $\hat{\tau}_{k}=[\boldsymbol{\vartheta}_{k}]_{1}+\tau_{l}$ and $\hat{\nu}_{k}=[\boldsymbol{\vartheta}_{k}]_{2}$, where $[\boldsymbol{\vartheta}_{k}]_{1}$ and $[\boldsymbol{\vartheta}_{k}]_{2}$ are the first and the second elements of $\boldsymbol{\vartheta}_{k}$, respectively.

Eventually, by iteratively conducting this process $K$ times, we can acquire the entire estimated delay and Doppler, which are $\hat{\boldsymbol{\tau}}=[\hat{\tau}_{1},\hat{\tau}_{2},\cdots,\hat{\tau}_{K}]^{\mathsf{T}}$ and $\hat{\boldsymbol{\nu}}=[\hat{\nu}_{1},\hat{\nu}_{2},\cdots,\hat{\nu}_{K}]^{\mathsf{T}}$.

\subsection{Lower Bound Analysis}\label{Sec-V-C}
We analyze the lower bound of the delay-Doppler estimation to evaluate the efficiency of the proposed \ac{MC-AFDM}.
Observation vector for the estimation is denoted as $\bar{\mathbf{y}}\triangleq[\bar{\mathbf{y}}^{(1)},\bar{\mathbf{y}}^{(2)},\cdots,\bar{\mathbf{y}}^{(P)}]^{\mathsf{T}}$, where $\bar{\mathbf{y}}^{(p)}\triangleq [\bar{\mathbf{y}}^{(p)}_{0},\bar{\mathbf{y}}^{(p)}_{1},\cdots,\bar{\mathbf{y}}^{(p)}_{N-1}]^{\mathsf{T}} \in \mathbb{R}^{NS_{p} \times 1}$ with sampled measurement $\bar{\mathbf{y}}_m^{(p)} \triangleq[\bar{y}_m^{(p)}[1], \bar{y}_m^{(p)}[2], \cdots, \bar{y}_m^{(p)}\left[S_p\right]]^\mathsf{T} \in \mathbb{R}^{S_p \times 1}$.
Thereafter, let the parameter vector $\boldsymbol{\eta}$ be defined as
\begin{equation}
\boldsymbol{\eta} \triangleq  {[ \boldsymbol{\vartheta}_1^\mathsf{T}, \boldsymbol{\vartheta}_2^\mathsf{T}, \cdots, \boldsymbol{\vartheta}_K^\mathsf{T}},  h_{l}, \boldsymbol{\varphi}]^{\mathsf{T}}, 
\end{equation}
where $\boldsymbol{\varphi} \in \mathbb{R}^{KNP \times 1}$ is the phase vector that stacks all of the phase $\varphi_{m, k}^{(p)}$ for $m=0,1,\cdots,N-1$ and $k=1,2,\cdots,K$, and $p=1,2,\cdots,P$.

Here, we are interested in the delay $\tau_{k}$ and Doppler $\nu_{k}$, where others are considered as nuisance parameters.
Letting $\boldsymbol{\vartheta}=[\boldsymbol{\vartheta}^{\mathsf{T}}_{1},\boldsymbol{\vartheta}^{\mathsf{T}}_{2},\cdots,\boldsymbol{\vartheta}^{\mathsf{T}}_{K}] \in \mathbb{R}^{2K \times 1}$ and $\boldsymbol{\chi}=[h_{l},\boldsymbol{\varphi}^{\mathsf{T}}]^{\mathsf{T}}$ denote parameter of interest vector and nuisance parameter vector so that $\boldsymbol{\eta}=[\boldsymbol{\vartheta}^{\mathsf{T}},\boldsymbol{\chi}^{\mathsf{T}}]^{\mathsf{T}}$.
The $(i,j)$-th element of corresponding \ac{FIM} of $\boldsymbol{\eta}$ is given by
\begin{align}\label{FIM}
    [\mathbf{J}(\boldsymbol{\eta})]_{(i, j)}&= \mathbb{E}\left\{\frac{\partial \ln p_{\bar{\mathbf{y}}}(\bar{\mathbf{y}} ; \boldsymbol{\eta})}{\partial[\boldsymbol{\eta}]_i} \frac{\partial \ln p_{\bar{\mathbf{y}}}(\bar{\mathbf{y}} ; \boldsymbol{\eta})}{\partial[\boldsymbol{\eta}]_j}\right\} \\&= \sum_{p=1}^P \sum_{m=0}^{N-1} \frac{1}{G_{T}} \int_{T^{(p-1)}}^{T^{(p)}} \frac{\partial \bar{y}_m^{(p)}(t)}{\partial[\boldsymbol{\eta}]_i} \frac{\partial \bar{y}_m^{(p)}(t)}{\partial[\boldsymbol{\eta}]_j}dt,\notag
\end{align}
where $\ln p_{\bar{\mathbf{y}}}(\bar{\mathbf{y}}; \boldsymbol{\eta})$ denotes the \ac{PDF} of the random vector $\bar{\mathbf{y}}$.

Note that the upper $2K \times 2K$ sub-matrix of $\mathbf{J}(\boldsymbol{\eta})$ is of our interest.
Therefore, we define \ac{EFIM}, which is a matrix with reduced dimension of the original \ac{FIM}, but retains the information to acquire the lower bound of specific parameters of interest~\cite {shen2010fundamental}.
Here, the \ac{FIM} in \eqref{FIM} is rewritten as
\begin{equation} \label{FIM_V2}
\mathbf{J}(\boldsymbol{\eta})=\left[\begin{array}{cc}
\mathbf{J}_{\boldsymbol{\vartheta}\boldsymbol{\vartheta}} & \mathbf{J}_{\boldsymbol{\vartheta}\boldsymbol{\chi}} \\
\mathbf{J}_{\boldsymbol{\vartheta}\boldsymbol{\chi}}^\mathsf{T} & \mathbf{J}_{\boldsymbol{\chi}\boldsymbol{\chi}}
\end{array}\right],
\end{equation}
where $\mathbf{J}_{\boldsymbol{\vartheta} \boldsymbol{\vartheta}} \in \mathbb{R}^{2K \times 2K}$, $\mathbf{J}_{\boldsymbol{\vartheta} \boldsymbol{\chi}} \in \mathbb{R}^{2K \times (KNP+1)}$, and $\mathbf{J}_{\boldsymbol{\chi} \boldsymbol{\chi}} \in \mathbb{R}^{(KNP+1) \times (KNP+1)}$ are block matrices, respectively.

Then, by taking the Schur complement of block matrix $\mathbf{J}_{\boldsymbol{\chi} \boldsymbol{\chi}}$, the \ac{EFIM} is given by~\cite{shen2010fundamental}
\begin{equation} \label{EFIM}
\mathbf{J}_{\mathrm{e}}(\boldsymbol{\vartheta})=\mathbf{J}_{\boldsymbol{\vartheta} \boldsymbol{\vartheta}}-\mathbf{J}_{\boldsymbol{\vartheta}\boldsymbol{\chi}}\mathbf{J}^{-1}_{\boldsymbol{\chi}\boldsymbol{\chi}}\mathbf{J}^{\mathsf{T}}_{\boldsymbol{\vartheta}\boldsymbol{\chi}},
\end{equation}
where $\mathbf{J}_{\mathrm{e}}(\boldsymbol{\vartheta})=\mathrm{blkdiag}\{\mathbf{J}_{\mathrm{e}}(\boldsymbol{\vartheta}_{1}),\mathbf{J}_{\mathrm{e}}(\boldsymbol{\vartheta}_{2}),\cdots,\mathbf{J}_{\mathrm{e}}(\boldsymbol{\vartheta}_{K})\}$.

Given this \ac{EFIM} and assuming the frame durations are equal (i.e., $T^{(p)}-T^{(p-1)}=N \Delta t$ for $p=1,2,\cdots,P$), the estimation accuracy of $\tau_{k}$ and $\nu_{k}$ are lower bounded by 
\begin{align}\label{Lower_Bound}
\mathrm{CRLB}(\tau_{k})&=[\mathbf{J}^{-1}_{\mathrm{e}}(\boldsymbol{\vartheta}_{k})]_{(1,1)} \approx \frac{1}{4 P \tilde{\varrho}_k \textrm{Var}\left(\tilde{\mathbf{c}}_1\right)},  \\ \mathrm{CRLB}(\nu_{k})&=[\mathbf{J}^{-1}_{\mathrm{e}}(\boldsymbol{\vartheta}_{k})]_{(2,2)} \approx \frac{\sum_{p=1}^P(\tilde{c}_1^{(p)})^2}{P^2 \tilde{\varrho}_k \operatorname{Var}\left(\tilde{\mathbf{c}}_1\right)}, \label{Lower_Bound_2}
\end{align}
where $\mathrm{CRLB}(\tau_{k})$ and $\mathrm{CRLB}(\nu_{k})$ are the lower bounds of $\tau_{k}$ and $\nu_{k}$, which are typically termed as \ac{CRLB}, respectively.

Here, $\tilde{\varrho}_k=\sum^{N-1}_{m=0}\tilde{\varrho}_{m, k}^{(p)}$ is the sum of information weight of the $k$-th target and $\mathrm{Var}(\mathbf{\tilde{c}}_{1})=\frac{1}{P} \sum_{p=1}^P(\tilde{c}_1^{(p)})^2-(\frac{1}{P} \sum_{p=1}^P \tilde{c}_1^{(p)})^2$ denotes the variance of chirp-rate vector $\mathbf{\tilde{c}}_{1}$.
The detailed derivation of the \ac{CRLB} is presented in Appendix~\ref{LB_Derivation}.
Based on the CRLB equations in \eqref{Lower_Bound}-\eqref{Lower_Bound_2}, the corresponding range and velocity lower bounds are given by
\begin{equation}\label{Lower_Bound_Range}
\mathrm{CRLB}(R_{k})= \frac{c^2}{16 P \tilde{\varrho}_k \textrm{Var}\left(\tilde{\mathbf{c}}_1\right)},
\end{equation}
\begin{equation}\label{Lower_Bound_Velocity}
\mathrm{CRLB}(v_{k})= \frac{\pi^2c^2\sum_{p=1}^P(\tilde{c}_1^{(p)})^2}{\omega^2_{\mathrm{RF}} P^2 \tilde{\varrho}_k \operatorname{Var}\left(\tilde{\mathbf{c}}_1\right)}.
\end{equation}

\section{Simulation Results}\label{Sec-VI}

\begin{table}[t]
\vspace{-2ex}
\rowcolors{1}{white}{blue!10!white}
\centering
\caption{Parameter setting for \ac{RAQRs}}  \label{table: simulation parameters}
\vspace{-0.5ex}
\resizebox{1\linewidth}{!}{
\begin{tabular}{|c|c|c|}
\hline
\textbf{Description} & \textbf{Parameter} & \textbf{Value}\\ \hline
Energy levels  & $\{|1\rangle, |2\rangle\}$ & $\{6S_{1/2},6P_{3/2}\}$ \\
Energy levels  & $\{|3\rangle, |4\rangle\}$ & $\{60D_{5/2},63P_{3/2}\}$ \\
Frequency & $\omega_{\mathrm{RF}}$ & $2\pi \times 62.76\,{\rm{GHz}}$ \\
Transition dipole moments & $\{\mu_{12},\mu_{34}\}$ & $\{2.586,229\}qa_{0}$ \\
Length of vapor cell & $L$ & $0.02\, \mathrm{m}$ \\
Rabi frequencies & $\{\Omega_{p},\Omega_{c}\}$ & $2\pi \times\{5.8,1\} \, \mathrm{MHz}$ \\
Wavelengths & $\{\lambda_{p},\lambda_{c}\}$ & $\{852,509\}\,{\rm{nm}}$ \\
Input power of probe beam & $P_{\mathrm{in}}$ & 120 $\mu\mathrm{W}$ \\
Decay rate & $\gamma_{2}$ & $2\pi \times 5.2\,\rm{MHz}$ \\
Quantum efficiency & $\eta$ & $0.8$ \\
Density of atoms & $N_{0}$ & $4.89 \times 10^{16}\, \mathrm{m^{-3}}$ \\
Bandwidth limitation of \ac{RAQRs} & $B_{\mathrm{RAQR}}$ & $10\,{\rm{MHz}}$ \\
\hline
\end{tabular}
}

\end{table}

\begin{table}[t]

\rowcolors{1}{white}{blue!10!white}
\centering
\caption{Parameter setting for \ac{MC-AFDM}}  \label{table: simulation parameters_2}
\vspace{-0.5ex}
\resizebox{1\linewidth}{!}{
\begin{tabular}{|c|c|c|}
\hline
\textbf{Description} & \textbf{Parameter} & \textbf{Value}\\ \hline
Number of chirp-subcarriers & $N$ & 8 \\
Number of post-chirps & $P$ & 16 \\
Sampling rate & $\Delta t$ & $10 \; \mathrm{ns}$  \\
Pre-chirp parameter & $c_{2}$ & $\sqrt{2}$ \\
Maximum chirp-rate & $\tilde{c}_{1}^{(\mathrm{max}) }$ & $4.65 \times 10^{12}$ \\
Minimum chirp-rate & $\tilde{c}_{1}^{(\mathrm{min}) }$ & $4.65 \times 10^{10}$ \\
Minimum chirp-rate difference & $\Delta \tilde{c}^{(\mathrm{min})}_{1}$ & $6.1 \times 10^{10}$ \\
\hline
\end{tabular}
}

\end{table}

\subsection{Simulation Environments}\label{Sec-VI-A}

Unless specified, the default setting of the simulations is organized in Tables~\ref{table: simulation parameters} and~\ref{table: simulation parameters_2}
.
The number of targets $K$ is set to 3.
The range and velocity of target are randomly distributed over $R_{k} \in [10,10^{3}]~\mathrm m$ and $v_{k} \in [50,300]~\mathrm{m/s}$, while the transmitter-to-\ac{RAQRs} distance is fixed to $1~\mathrm{m}$.
The fluctuation frequency grid point for \ac{OMP} $G_{f}$ is set to $2^{13}$.
The average received \ac{SNR} is defined as 
\begin{equation} \label{SNR_Definition}
\mathrm{SNR}(t) \triangleq \frac{1}{N  P} \sum_{p=1}^P \sum_{m=0}^{N-1} (\varrho_{m}^{(p)}(t))^2  h_l^2 N\Delta t,
\end{equation}
where $\varrho_m^{(p)}(t)=\frac{1}{K} \sum_{k=1}^K \varrho_{m, k}^{(p)}(t)$ is the average amplitude for $K$ targets at the $m$-th subcarrier during the $p$-th frame.
Note that the received \ac{SNR} remains constant when the transmission power is fixed \cite{cui2025realizing}.
For the estimation performance analysis, we adopted the \ac{NRMSE}, which is defined as
\begin{equation} \label{NRMSE}
\mathrm{NRMSE}=\mathbb{E}\left\{\sqrt{\frac{1}{K}\sum_{k=1}^{K}\frac{\left(\hat{\zeta}_{k}-\zeta_{k}\right)^2}{|\zeta_{k}|^2}}\; \right\},
\end{equation}
where $\zeta_{k}$ and $\hat{\zeta}_{k}$ denote a channel parameter and its estimate of the $k$-th target, which are selected as range and velocity (i.e., $\zeta_{k} \in \{R_{k},v_{k}\}$), respectively.

\begin{figure}[t]
\captionsetup[subfigure]{justification=centering}
\begin{center} 
\begin{subfigure}[t]{1\columnwidth}
\includegraphics[width=1\columnwidth]{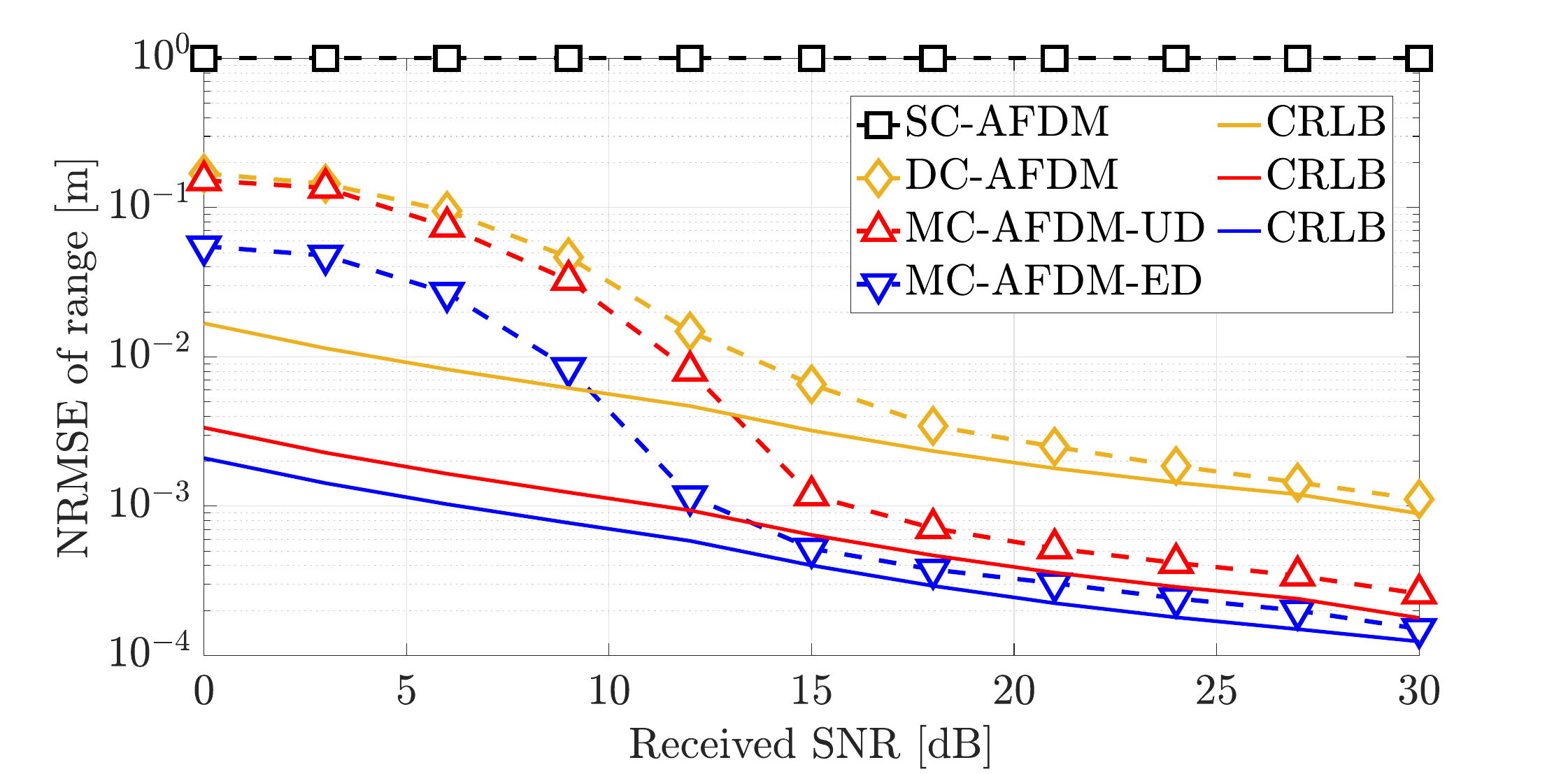}
\caption{\ac{NRMSE} of range according to the received SNR.}
\label{NRMSE_SNR_R}
\end{subfigure}
\begin{subfigure}[t]{1\columnwidth}
\includegraphics[width=1\columnwidth]{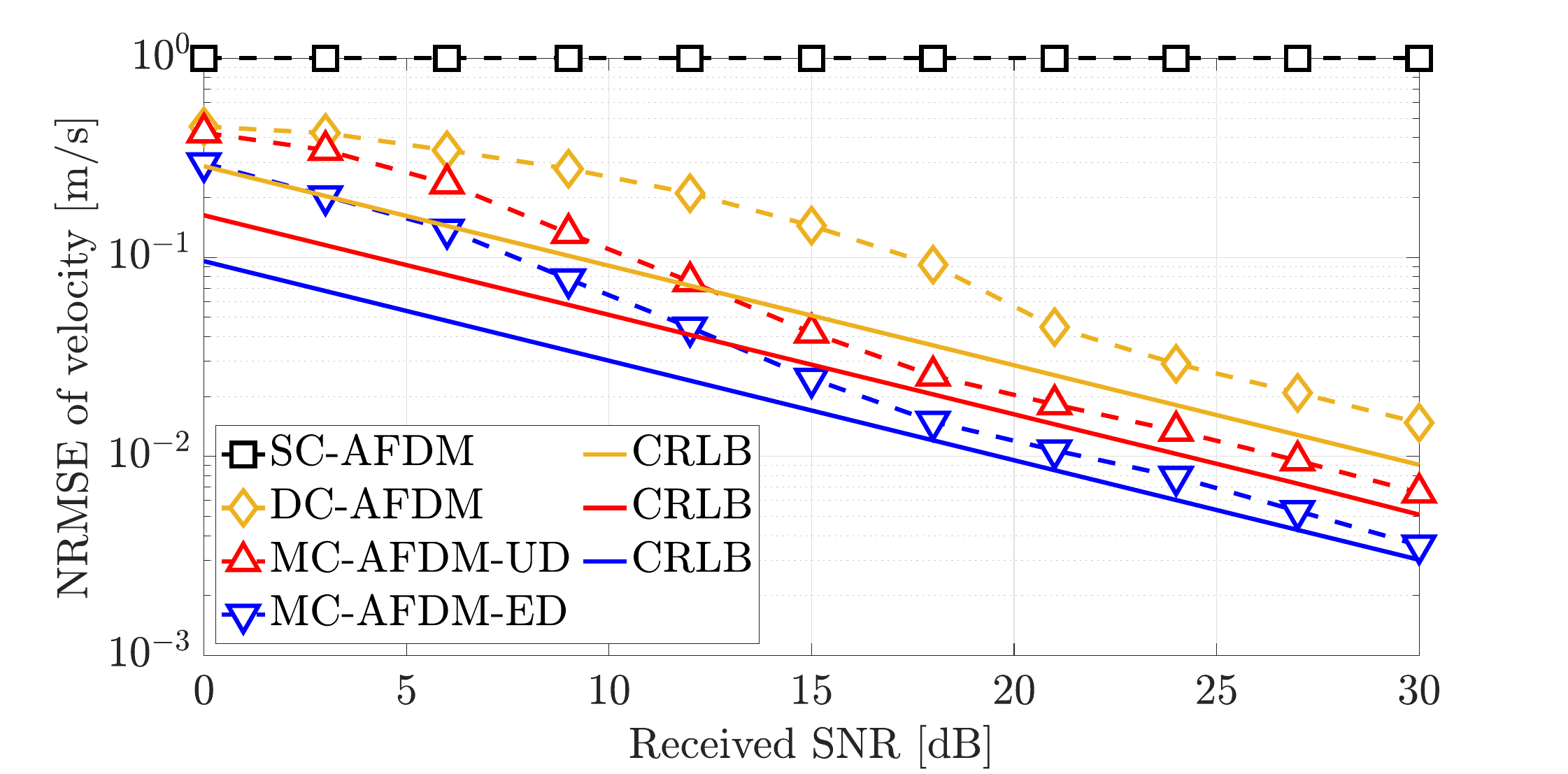}
\caption{\ac{NRMSE} of velocity according to the received SNR.}
\label{NRMSE_SNR_V}
\end{subfigure}
\caption{Range-velocity estimation \ac{NRMSE} performance analysis versus received \ac{SNR}.}\label{NRMSE_SNR} 
\end{center}
\vspace{-3.8ex}
\end{figure}

\subsection{NRMSE Performance Analysis}\label{Sec-VI-B}
For the \ac{NRMSE} performance analysis, the proposed \ac{MC-AFDM} is evaluated with two different post-chirp distributions.
The first is uniform distribution of chirp-rate vector elements (i.e., $\tilde{c}_1^{(p)}=\tilde{c}_1^{(\min )}+\frac{p(\tilde{c}_1^{(\max )}-\tilde{c}_1^{(\min )})}{P}$) for $p=1,2,\cdots,P$, which is represented as \textit{MC-AFDM-UD}.
The second is the edge distribution that follows Theorem~\ref{Theorem_1} and adopts the minimum chirp-rate difference $\Delta \tilde{c}^{(\mathrm{min})}_{1}$ in Table~\ref{table: simulation parameters_2}, which is labeled \textit{MC-AFDM-ED}.
Three benchmark schemes are considered in the performance comparison, which are
%

\begin{itemize}
   \item \textit{SC-AFDM}: Conventional \ac{SC-AFDM} is established by assuming the post chirp difference for \ac{SC-AFDM} as $\Delta c^{(SC)}_{1}=c^{(A)}_{1}-c^{(B)}_{1} \approx 0$ since the identical chirp rate (i.e., $\Delta c^{(SC)}_{1}=0$) results in a rank-1 matrix of $\mathbf{C}_{1}$, which limits delay-Doppler estimation of \ac{SC-AFDM}.
   Here, the post chirp difference for \ac{SC-AFDM} is $\Delta c^{(SC)}_{1}=10^{-3}$.
   \item \textit{DC-AFDM}: \ac{DC-AFDM} is a simplified version of \ac{MC-AFDM}, which exploits only two different post-chirps~\cite{kim2026dual}. 
   \item \textit{CRLB}: \ac{CRLB} of range and velocity estimation, derived in \eqref{Lower_Bound_Range} and \eqref{Lower_Bound_Velocity}, are adopted to evaluate the theoretical lower bounds.
   The \ac{CRLB}s are calculated for \ac{DC-AFDM}, MC-AFDM-UD, and MC-AFDM-ED, where their respective CRLBs are presented as solid lines.
   %
\end{itemize}
%
%

Fig.~\ref{NRMSE_SNR} shows the NRMSE performance versus the received \ac{SNR}.
As illustrated in Fig.~\ref{NRMSE_SNR}, the proposed \ac{MC-AFDM} exhibits superior estimation accuracy in both range and velocity for all received \ac{SNR}s.
Especially, an interesting observation is that the estimation performance of MC-AFDM-ED is more accurate compared to the MC-AFDM-UD in all received SNRs, showing the effectiveness of the edge distribution of chirp-rates.
Furthermore, this result shows that the multiple post-chirps of \ac{MC-AFDM} reduce the condition number $\kappa(\mathbf{C_{1}})$ in \eqref{condition_number}, thereby enhancing the estimation accuracy compared to \ac{DC-AFDM}.
The NRMSE of \ac{SC-AFDM} is the worst since the single-chirp cannot solve the optical ambiguity problem, severely degrading the estimation performance.

\begin{figure}[t]
\captionsetup[subfigure]{justification=centering}
\begin{center} 
\begin{subfigure}[t]{1\columnwidth}
\includegraphics[width=1\columnwidth]{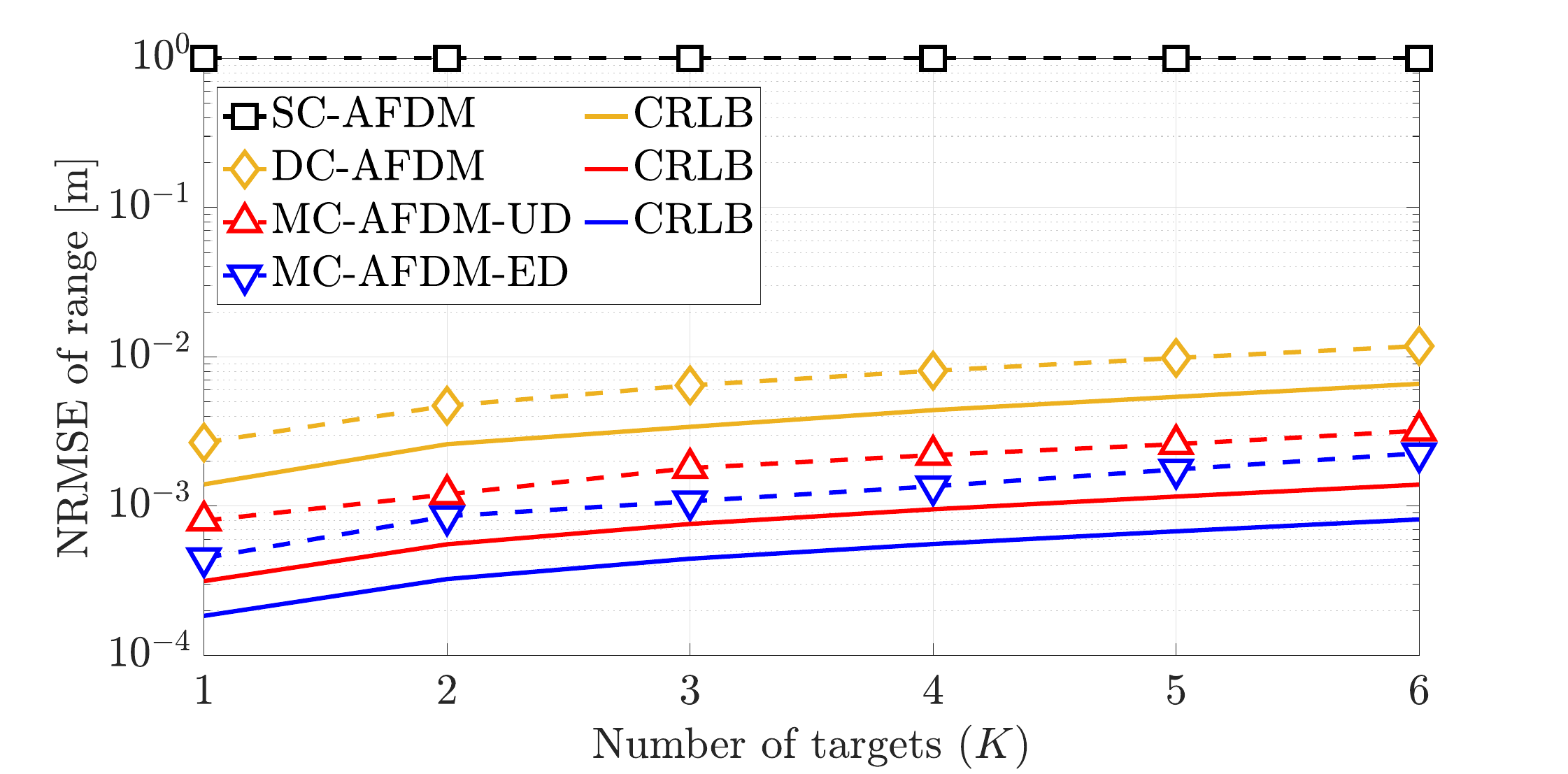}
\caption{\ac{NRMSE} of range according to the number of targets.}
\label{NRMSE_Target_R}
\end{subfigure}
\begin{subfigure}[t]{1\columnwidth}
\includegraphics[width=1\columnwidth]{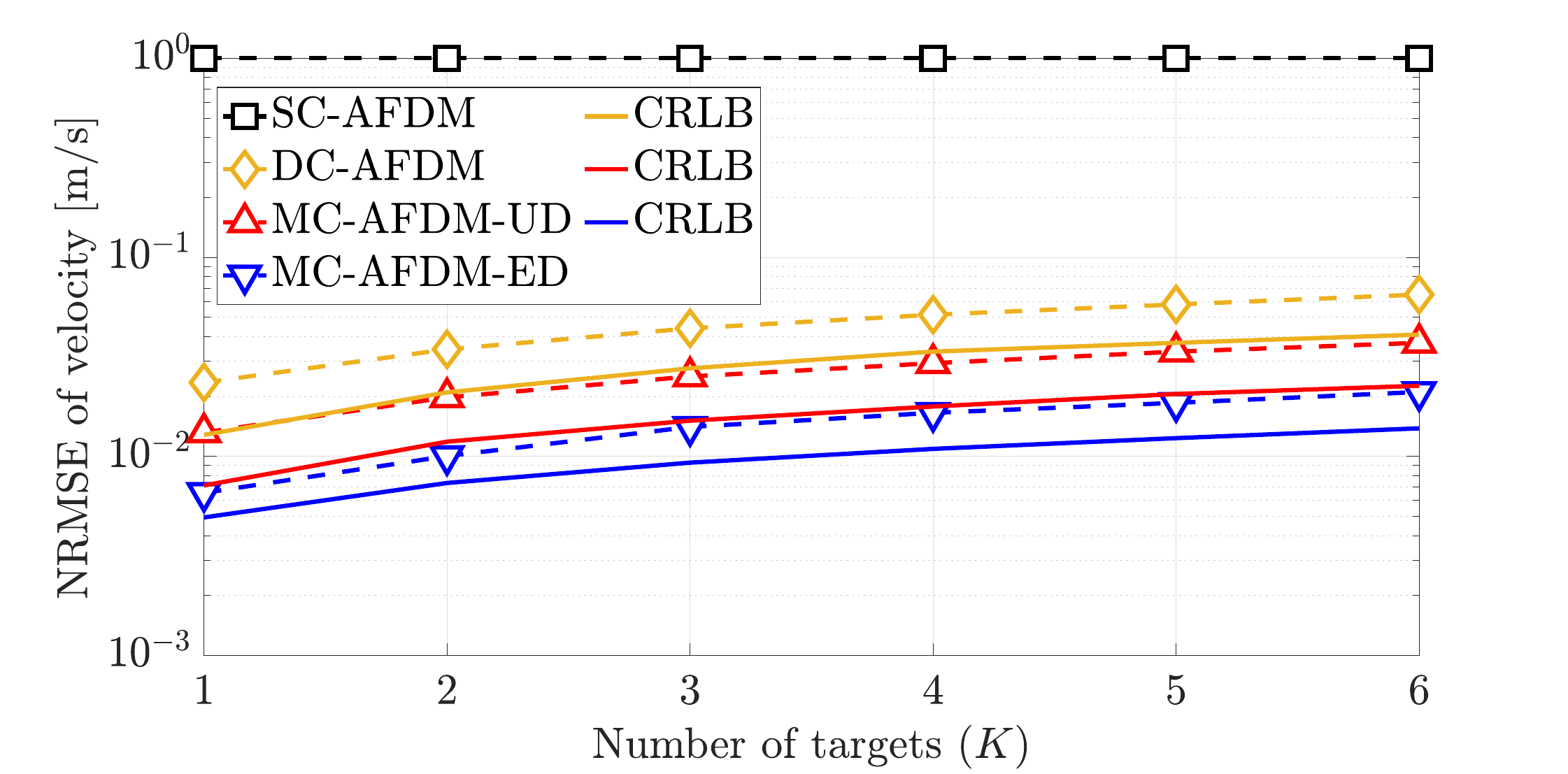}
\caption{\ac{NRMSE} of velocity according to the number of targets.}
\label{NRMSE_Target_V}
\end{subfigure}
\caption{Range-velocity estimation \ac{NRMSE} performance analysis versus the number of targets.}\label{NRMSE_Target} 
\end{center}
\vspace{-3.8ex}
\end{figure}

To analyze the delay and Doppler estimation performance with respect to the number of targets, the \ac{NRMSE} result of range and velocity versus the number of targets is shown in Fig.~\ref{NRMSE_Target}.
Here, the received \ac{SNR} is set to 20 dB.
As shown in Fig.~\ref{NRMSE_Target}, the proposed \ac{MC-AFDM} outperforms all the benchmark schemes for all numbers of targets, validating the effectiveness in the multi-target scenario.
In particular, the edge distribution of post-chirp rates reduces the \ac{NRMSE} of range and velocity compared to the uniform distribution, on the order of 2.5-fold and 4-fold, respectively, which validates the advantage of edge distribution chirp-matrix design for delay-Doppler estimation in the multi-target scenario. 

\begin{figure}[t]
\captionsetup[subfigure]{justification=centering}
\begin{center} 
\begin{subfigure}[t]{1\columnwidth}
\includegraphics[width=1\columnwidth]{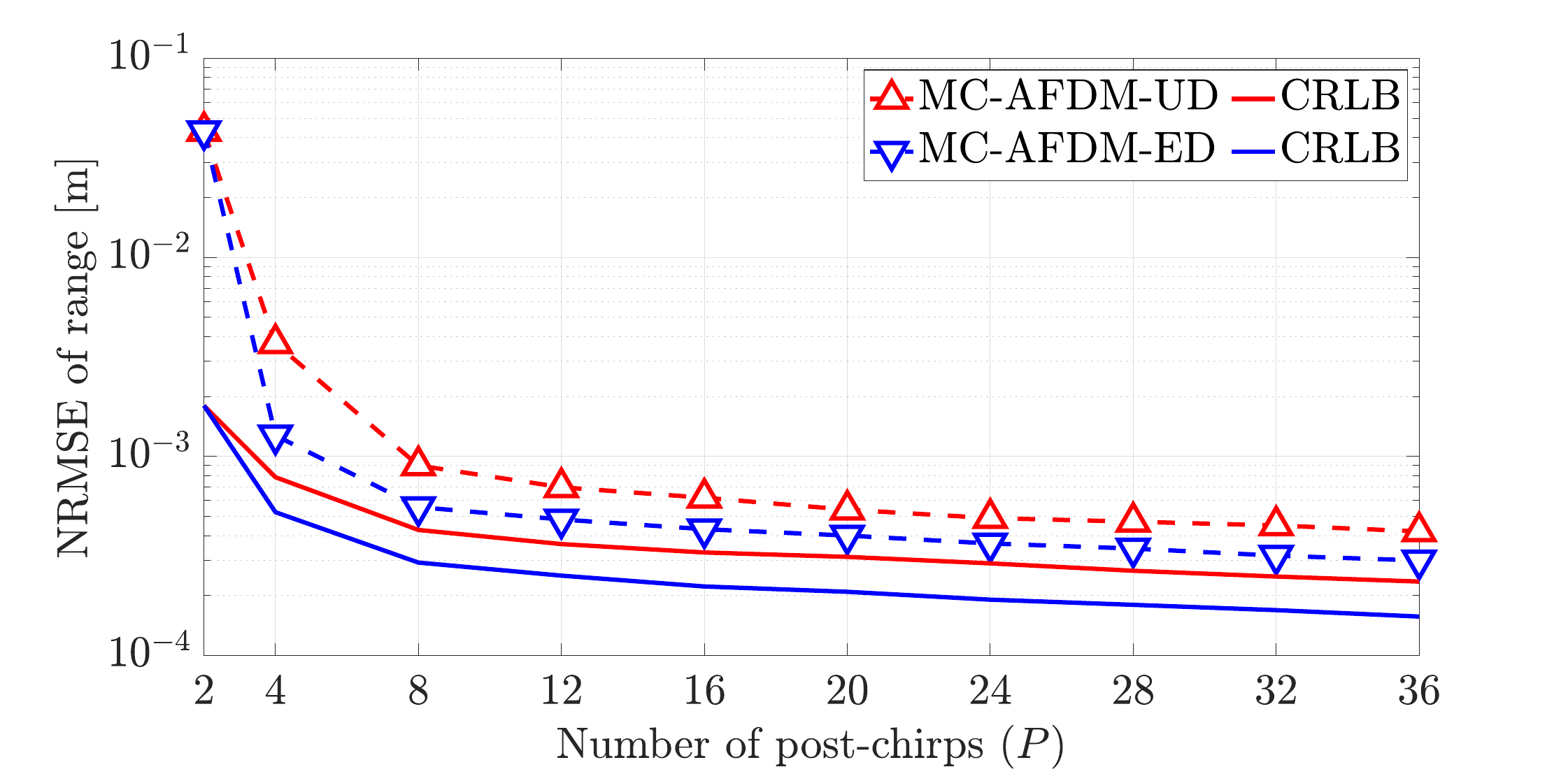}
\caption{\ac{NRMSE} of range according to the number of post-chirps.}
\label{NRMSE_Chirp_R}
\end{subfigure}
\begin{subfigure}[t]{1\columnwidth}
\includegraphics[width=1\columnwidth]{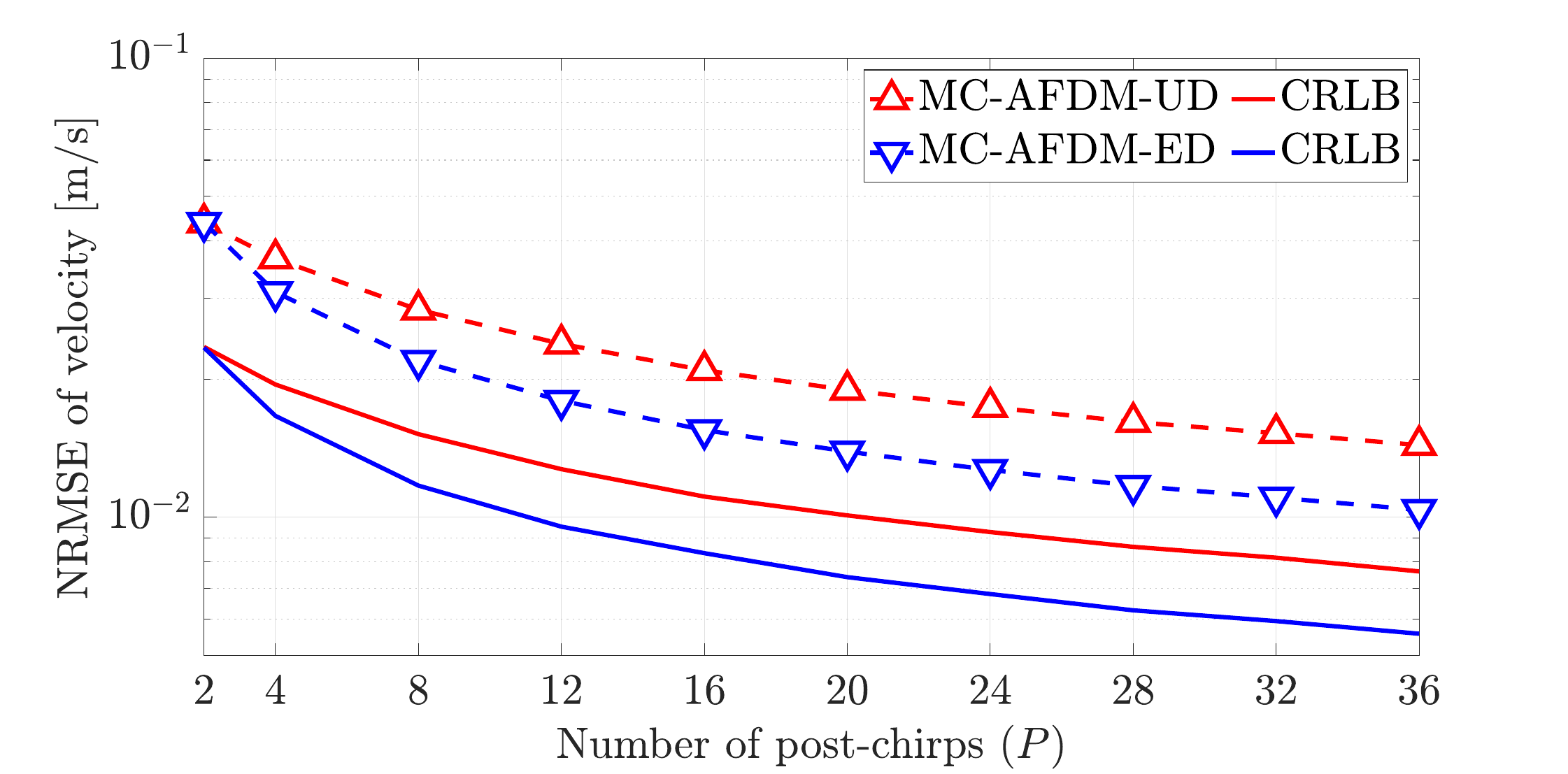}
\caption{\ac{NRMSE} of velocity according to the number of post-chirps.}
\label{NRMSE_Chirp_V}
\end{subfigure}
\caption{Range-velocity estimation \ac{NRMSE} performance analysis versus the number of post-chirps.}\label{NRMSE_Chirp} 
\end{center}
\vspace{-3.8ex}
\end{figure}

Eventually, Fig.~\ref{NRMSE_Chirp} shows the \ac{NRMSE} performance according to the varying number of post-chirps.
Note that \ac{SC-AFDM} and \ac{DC-AFDM} are excluded in this analysis since their number of post-chirps is fixed as one and two, respectively.
The received \ac{SNR} is set to 20 dB.
The experimental results show that both range and velocity estimation performance increase along with the number of post-chirps.
This is because the objective parameter $D_{2}$ in Appendix~\ref{app2} is proportional to $P$, which reduces the condition number $\kappa\left(\mathbf{C}_1\right)$ as $P$ increases.
Furthermore, as can be observed in \eqref{Lower_Bound_Range} and \eqref{Lower_Bound_Velocity}, increasing $P$ directly lowers the theoretical estimation lower bound, indicating that a larger number of post-chirps provides more informative measurements for the LS recovery.

In summary, these results demonstrate the effectiveness of the proposed \ac{MC-AFDM} for delay-Doppler estimation with \ac{RAQRs} in the multi-target scenario.
Especially, by distributing the multiple post-chirps to the edge of its minimum and maximum values and increasing its numbers, the delay-Doppler estimation accuracy can be further improved, unveiling the potential of \ac{RAQRs} for quantum sensing in a \ac{DD} channel.

\section{Conclusion}\label{Sec-VII}
In this paper, we proposed \ac{MC-AFDM} for joint delay-Doppler estimation in the \ac{DD} channel with \ac{RAQRs}.
Leveraging multiple post-chirps within different time frames, the proposed \ac{MC-AFDM} generates diverse optical measurements that resolve the optical ambiguity of conventional \ac{SC-AFDM}, ensuring the full-rank condition of the post-chirp matrix.
Furthermore, we have revealed that the delay-Doppler estimation accuracy could be further improved when the variance of the chirp rate vector is maximized, which can be achieved by distributing the post-chirps to the edges of their maximum and minimum values.
Numerical results validated the superior delay-Doppler estimation performance of \ac{MC-AFDM} in the multi-target scenarios compared to \ac{SC-AFDM} and \ac{DC-AFDM}, and have shown the advantage of edge distribution of post-chirps for enhancing the accuracy.

This study addressed a new research branch of \ac{RAQRs}, which is waveform design for \ac{RAQRs}-enabled quantum sensing.
The proposed \ac{MC-AFDM} can be extended to several promising research avenues, including \ac{AFDM}-based \ac{ISAC} with \ac{RAQRs}, channel estimation, and integration with \ac{MIMO} systems.
Furthermore, adopting the multi-band capability of \ac{RAQRs} with \ac{AFDM} can be considered to enhance the spectral efficiency of communications.

\begin{appendices}

\section{Proof of Theorem 1} \label{app2}
The positive definite matrix $\mathbf{C}_1^{\mathsf{T}} \mathbf{C}_1$ is given by
\begin{equation} \label{pd_matrix}
\mathbf{C}_1^{\mathsf{T}} \mathbf{C}_1=\left[\begin{array}{cc}
4 \sum_{p=1}^P (\tilde{c}_1^{(p)})^2 & -2 \sum_{p=1}^P \tilde{c}_1^{(p)} \\
-2 \sum_{p=1}^P \tilde{c}_1^{(p)} & P
\end{array}\right].
\end{equation}
Here, the eigenvalues of $\mathbf{C}_1^{\mathsf{T}} \mathbf{C}_1$, $\lambda_{1}$ and $\lambda_{2}$, can be calculated by eigenvalue-decomposition, which are presented as
\begin{equation} \label{eigenvalues}
\lambda_{1}=\frac{D_{1} + \sqrt{D_{1}^2-4 D_{2}}}{2}, \ \lambda_{2}=\frac{D_{1} - \sqrt{D_{1}^2-4 D_{2}}}{2},
\end{equation}
where $D_{1}$ and $D_{2}$ are given by
\begin{equation}\label{c1c2}
D_{1}=4 \sum_{p=1}^P (\tilde{c}_1^{(p)})^2+P,  
\end{equation}
\begin{equation}\label{c1c2_2}
D_{2}=4\left(P \sum_{p=1}^P (\tilde{c}_1^{(p)})^2-\left(\sum_{p=1}^P \tilde{c}_1^{(p)}\right)^2\right). 
\end{equation}
%
Given $D_{1}=\lambda_{1}+\lambda_{2}$ and $D_{2}=\lambda_{1}\lambda_{2}$, we formulate the equation as follows
\begin{equation} \label{opt_obj}
\left(\kappa(\mathbf{C}_{1})+(\kappa(\mathbf{C}_{1}))^{-1}\right)^2=\frac{\left(\lambda_1+\lambda_2\right)^2}{\lambda_1 \lambda_2}=\frac{D_1^2}{D_2}.
\end{equation}
Since $\kappa(\mathbf{C}_{1})+(\kappa(\mathbf{C}_{1}))^{-1}$ is monotonically increasing with $\kappa(\mathbf{C}_{1})$ for $\kappa(\mathbf{C}_{1}) \geq 1$, the objective is simplified as
\begin{equation} \label{opt_obj_v2}
\min \kappa\left(\mathbf{C}_1\right) \Longleftrightarrow \min \frac{D_1}{\sqrt{D_2}} .
\end{equation}
%
%
Note that $D_{2}$ can be represented as $D_{2}=4 P^2 \mathrm{Var}(\mathbf{\tilde{c}}_{1})$.
Considering a fixed number of post-chirp $P$, the optimization problem in \eqref{opt_obj_v2} can be solved by maximizing the $\mathrm{Var}(\mathbf{\tilde{c}}_{1})$, which can be achieved by spreading the elements of $\mathbf{\tilde{c}}_{1}$ as far as possible within constraints for \ac{MC-AFDM}.

\section{Lower Bound Analysis} \label{LB_Derivation}
The partial derivatives of $\bar{y}^{(p)}_{m}(t)$ with respect to parameters in $\boldsymbol{\eta}$ are given by
\begin{equation}\label{Partial_1}
\frac{\partial \bar{y}_m^{(p)}(t)}{\partial[\boldsymbol{\vartheta}_{k}]_{1}}=-2 \tilde{c}_1^{(p)} t \varrho_{m, k}^{(p)}(t) h_l \sin (\omega_{m, k}^{(p)} t+\varphi_{m, k}^{(p)}),
\end{equation}
\begin{equation}\label{Partial_2}
\frac{\partial \bar{y}_m^{(p)}(t)}{\partial[\boldsymbol{\vartheta}_{k}]_{2}}=t \varrho_{m, k}^{(p)}(t) h_l \sin (\omega_{m, k}^{(p)} t+\varphi_{m, k}^{(p)}),
\end{equation}
\begin{equation}\label{Partial_3}
\frac{\partial \bar{y}_m^{(p)}(t)}{\partial h_{l}}=\sum_{k=1}^K \varrho_{m, k}^{(p)}(t) \cos (\omega_{m, k}^{(p)} t+\varphi_{m, k}^{(p)}),
\end{equation}
\begin{equation}\label{Partial_4}
\frac{\partial \bar{y}_m^{(p)}(t)}{\partial\varphi_{m, k}^{(p)}}=-\varrho_{m, k}^{(p)}(t) h_l \sin (\omega_{m, k}^{(p)} t+\varphi_{m, k}^{(p)}).
\end{equation}
%

In the self-heterodyne sensing system, the fluctuation frequency $\omega^{(p)}_{m,k}$ is typically very large compared to the time-varying amplitude $\varrho_{m, k}^{(p)}(t)$ \cite{cui2025realizing}.
Owing to this characteristic, the following asymptotic approximations hold
\small\begin{equation}\label{approx_1}
\int_{T^{(p-1)}}^{T^{(p)}} (\varrho_{m, k}^{(p)}(t))^2 \sin ^2(g^{(p)}_{m,k}(t)) d t \approx \frac{1}{2} \int_{T^{(p-1)}}^{T^{(p)}} (\varrho_{m, k}^{(p)}(t))^2 d t, 
\end{equation}
\small\begin{equation}\label{approx_2}
\int_{T^{(p-1)}}^{T^{(p)}} (\varrho_{m, k}^{(p)}(t))^2 \sin (g^{(p)}_{m,k}(t)) \cos (g^{(p)}_{m,k}(t)) d t \approx 0,
\end{equation}
\small\begin{equation}\label{approx_3}
\int_{T^{(p-1)}}^{T^{(p)}} \varrho_{m, k}^{(p)} \varrho_{m, k^{\prime}}^{(p)} f_{m, k}^{(p)}(t)  f_{m, k^{\prime}}^{(p)}(t) d t \approx 0,\quad k \neq k^{\prime} 
\end{equation}
where $f_{m, k}^{(p)}(t) \in\{\sin (g^{(p)}_{m,k}(t)), \cos (g^{(p)}_{m,k}(t))\}$ and $g^{(p)}_{m,k}(t)=\omega^{(p)}_{m,k}t+\varphi^{(p)}_{m,k}$.
Given the above, by substituting \eqref{Partial_1}-\eqref{Partial_4} into \eqref{FIM} and applying the approximations in \eqref{approx_1}-\eqref{approx_3}, the block matrices of \ac{FIM} in \eqref{FIM_V2} are calculated as follows
\begin{equation}\label{FIM_V4} 
\left[\mathbf{J}_{\vartheta \vartheta}\right]_{(k, k)}=\sum_{p=1}^P \sum_{m=0}^{N-1} \tilde{\varrho}^{(p)}_{m,k,2}\left[\begin{array}{cc}
(2 \tilde{c}_1^{(p)})^2 & -2 \tilde{c}_1^{(p)} \\
-2 \tilde{c}_1^{(p)} & 1
\end{array}\right], 
\end{equation}
\begin{equation}\label{FIM_V5} 
\mathbf{J}_{\chi \chi}=\left[\begin{array}{ll}
\sum_{k=1}^K \sum_{p=1}^P \sum_{m=0}^{N-1} \frac{\tilde{\varrho}_{m, k,0}^{(p)}}{h_l^2} & \boldsymbol{0}^{\mathsf{T}}_{KNP,1} \\
\boldsymbol{0}_{KNP, 1} & \mathbf{J}_{\varphi \varphi}
\end{array}\right],
\end{equation}
\begin{equation}\label{FIM_V6} 
\mathbf{J}_{\vartheta \chi}=\left[\boldsymbol{0}_{KNP, 1}, \mathrm{blkdiag}\{\mathbf{J}_{\vartheta \varphi,k}\}\right],
\end{equation}
where $\mathbf{J}_{\varphi \varphi}=\operatorname{diag}(\tilde{\varrho}_{m, k,0}^{(p)})^{(N-1,K,P{})}_{m=0, k=1, p=1} \in \mathbb{R}^{KNP \times KNP}$. The entries of the sub-block matrix $\mathbf{J}_{\vartheta \varphi,k} \in \mathbb{R}^{2 \times NP}$ are given by $\left[\mathbf{J}_{\vartheta \varphi, k}\right]_{(1,m p)}=2 \tilde{c}_1^{(p)} \tilde{\varrho}_{m, k,1}^{(p)}$ and $\left[\mathbf{J}_{\vartheta \varphi, k}\right]_{(2,m p)}=-\tilde{\varrho}_{m,k,1}^{(p)}$, respectively.
Here, parameters $\tilde{\varrho}_{m, k,0}^{(p)}$,$\tilde{\varrho}_{m, k,1}^{(p)}$, and $\tilde{\varrho}_{m, k,2}^{(p)}$ are given by
\begin{equation}\label{rho_0} 
\tilde{\varrho}_{m, k,0}^{(p)} \triangleq \frac{h_l^2}{2 G_T} \int_{T^{(p-1)}}^{T^{(p)}}(\varrho_{m, k}^{(p)}(t))^2 dt,
\end{equation}
\begin{equation}\label{rho_1} 
\tilde{\varrho}_{m, k,1}^{(p)} \triangleq \frac{h_l^2}{2 G_T} \int_{T^{(p-1)}}^{T^{(p)}}t(\varrho_{m, k}^{(p)}(t))^2 dt,
\end{equation}
\begin{equation}\label{rho_2} 
\tilde{\varrho}_{m, k,2}^{(p)} \triangleq \frac{h_l^2}{2 G_T} \int_{T^{(p-1)}}^{T^{(p)}} t^2 (\varrho_{m, k}^{(p)}(t))^2 dt.
\end{equation}

By substituting \eqref{FIM_V4}-\eqref{FIM_V6} and conducting a few calculations, the \ac{EFIM} in \eqref{EFIM} is obtained as $\mathbf{J}_{\mathrm{e}}(\boldsymbol{\vartheta})=\mathrm{blkdiag}\{\mathbf{J}_{\mathrm{e}}(\boldsymbol{\vartheta}_{1}),\mathbf{J}_{\mathrm{e}}(\boldsymbol{\vartheta}_{2}),\cdots,\mathbf{J}_{\mathrm{e}}(\boldsymbol{\vartheta}_{K})\}$, where the sub-matrix $\mathbf{J}_{\mathrm{e}}(\boldsymbol{\vartheta}_{k})$ is given by
\begin{equation} \label{EFIM_Fin}
\mathbf{J}_e\left(\boldsymbol{\vartheta}_k\right)=\sum_{p=1}^P \sum_{m=0}^{N-1} \tilde{\varrho}_{m, k}^{(p)}\left[\begin{array}{cc}
\left(2 \tilde{c}_1^{(p)}\right)^2 & -2 \tilde{c}_1^{(p)} \\
-2 \tilde{c}_1^{(p)} & 1
\end{array}\right],
\end{equation}
where $\tilde{\varrho}_{m, k}^{(p)}=\tilde{\varrho}_{m, k,2}^{(p)}-\frac{(\tilde{\varrho}_{m, k,1}^{(p)})^2}{\tilde{\varrho}_{m, k,0}^{(p)}}$ denotes the information weight of the \ac{EFIM}.
Here, \ac{EFIM} in \eqref{EFIM_Fin} can be rewritten with the post-chirp matrix $\mathbf{C}_1$ in \eqref{Fluctuation frequency}, which is presented as $\mathbf{J}_e(\boldsymbol{\vartheta}_k)=\mathbf{C}_1^\mathsf{T} \mathbf{D}_k \mathbf{C}_1$,
where $\mathbf{D}_k=\mathrm{diag}\{\tilde{\varrho}_k^{(1)}, \tilde{\varrho}_k^{(2)} \cdots, \tilde{\varrho}_k^{(P)}\} \in \mathbb{R}^{P \times P}$ denotes the diagonal matrix and its diagonal element is the sum of information weights $\tilde{\varrho}_k^{(p)}=\sum^{N-1}_{m=0}\tilde{\varrho}_{m, k}^{(p)}$.
Considering the uniform frame duration assumption, the sum of information weight is equivalent for all $p$ (i.e., $\tilde{\varrho}_k^{(1)}=\tilde{\varrho}_k^{(2)}=\cdots=\tilde{\varrho}_k^{(P)}$).
Therefore, we omit the subscript $p$ so that $\tilde{\varrho}_k^{(p)}=\tilde{\varrho}_k \;\forall\; p$.
Then, by taking the inverse of $\mathbf{J}_e\left(\boldsymbol{\vartheta}_k\right)$, the \ac{CRLB} of $\tau_{k}$ and $\nu_{k}$ are given by 
\begin{align}\label{Lower_Bound_Fin}
\mathrm{CRLB}(\tau_{k})&=[\mathbf{J}^{-1}_{\mathrm{e}}(\boldsymbol{\vartheta}_{k})]_{(1,1)}=\frac{P \tilde{\varrho}_k}{\mathrm{det}\left(\mathbf{C}_1^\mathsf{T} \mathbf{D}_k \mathbf{C}_1\right)},  \\ \mathrm{CRLB}(\nu_{k})&=[\mathbf{J}^{-1}_{\mathrm{e}}(\boldsymbol{\vartheta}_{k})]_{(2,2)}=\frac{4 \tilde{\varrho}_k\sum^{P}_{p=1}(\tilde{c}^{(p)}_{1})^2}{\mathrm{det}\left(\mathbf{C}_1^\mathsf{T} \mathbf{D}_k \mathbf{C}_1\right)},
\end{align}
where $\mathrm{det}\left(\mathbf{C}_1^\mathsf{T} \mathbf{D}_k \mathbf{C}_1\right)=4P^2\tilde{\varrho}^2_k\mathrm{Var}(\tilde{\mathbf{c}}_{1})$.
This completes the derivation of lower bound.

\end{appendices}


\end{document}